\documentclass[11pt,letterpaper]{article}
\usepackage[T1]{fontenc}
\usepackage[english]{babel}
\usepackage[margin=1in]{geometry}
\usepackage{lmodern,amsmath,amssymb,amsthm,mathtools,aliascnt}
\usepackage{booktabs,array,tabularx,enumitem,microtype,placeins}
\usepackage{tikz}
\usetikzlibrary{arrows.meta,positioning,fit,calc}
\usepackage{csquotes,xurl}
\usepackage[backend=biber,style=alphabetic,sorting=nyt,giveninits=false,
 maxbibnames=99,maxalphanames=4,minalphanames=4,
 doi=true,url=true,eprint=true,isbn=false]{biblatex}
\DeclareNameAlias{sortname}{given-family}
\DeclareFieldFormat[article]{eid}{Art.\addnbspace#1}
\renewbibmacro{in:}{\ifentrytype{article}{}{\printtext{\bibstring{in}\intitlepunct}}}
\AtBeginBibliography{\interlinepenalty=10000}
\usepackage[hidelinks]{hyperref}
\hypersetup{pdftitle={The Complexity of Exact and Approximate Nash Equilibria in Anonymous Games},
 pdfauthor={Zhengyang Liu},
 pdfsubject={Computational complexity of equilibria in anonymous games},
 pdfkeywords={anonymous games, Nash equilibrium, PPAD, FIXP, interpolation}}
\usepackage[nameinlink,capitalise,noabbrev]{cleveref}
\allowdisplaybreaks
\newtheorem{theorem}{Theorem}[section]
\newaliascnt{lemma}{theorem}
\newtheorem{lemma}[lemma]{Lemma}
\aliascntresetthe{lemma}
\newaliascnt{proposition}{theorem}

\aliascntresetthe{proposition}
\newaliascnt{corollary}{theorem}
\newtheorem{corollary}[corollary]{Corollary}
\aliascntresetthe{corollary}
\theoremstyle{definition}
\newaliascnt{definition}{theorem}
\newtheorem{definition}[definition]{Definition}
\aliascntresetthe{definition}
\newaliascnt{remark}{theorem}

\aliascntresetthe{remark}
\crefname{theorem}{Theorem}{Theorems}
\crefname{lemma}{Lemma}{Lemmas}
\crefname{proposition}{Proposition}{Propositions}
\crefname{corollary}{Corollary}{Corollaries}
\crefname{definition}{Definition}{Definitions}
\crefname{section}{Section}{Sections}
\crefname{subsection}{Section}{Sections}
\crefname{subsubsection}{Section}{Sections}
\crefname{equation}{Equation}{Equations}
\crefname{table}{Table}{Tables}
\crefname{figure}{Figure}{Figures}
\newcommand{\PPAD}{\mathsf{PPAD}}
\newcommand{\FIXP}{\mathsf{FIXP}}
\newcommand{\FIXPa}{\mathsf{FIXP}_a}
\newcommand{\PureCircuit}{\textnormal{\textsc{Pure-Circuit}}}
\newcommand{\GCircuit}{\textnormal{\textsc{Gcircuit}}}
\newcommand{\ExactGCircuit}{\textnormal{\textsc{Exact-Gcircuit}}}
\newcommand{\OQE}{\textnormal{\textsc{Ordered Quadratic Equations}}}
\newcommand{\eps}{\varepsilon}
\newcommand{\D}{\mathcal D}

\newcommand{\step}[1]{\par\smallskip\noindent\textbf{#1}\enspace}

\newcommand{\ANE}{\mathrm{ANE}}
\newcommand{\WSNE}{\mathrm{WSNE}}
\newcommand{\E}{\mathbb E}
\newcommand{\R}{\mathbb R}
\newcommand{\one}{\mathbf 1}

\newcommand{\calA}{\mathcal A}
\newcommand{\calR}{\mathcal R}
\newcommand{\poly}{\operatorname{poly}}

\newcommand{\clip}{\operatorname{clip}}

\DeclareMathOperator{\NOR}{NOR}
\DeclareMathOperator{\PURIFY}{PURIFY}
\title{The Complexity of Exact and Approximate Nash Equilibria in Anonymous Games}
\author{Zhengyang Liu\\
{\normalsize Beijing Institute of Technology}\\
{\normalsize\texttt{zhengyang@bit.edu.cn}}}
\date{}

\begin{document}
\pagenumbering{roman}
\begin{titlepage}
\thispagestyle{empty}
\begingroup
\centering
\makeatletter
{\LARGE\@title\par}
\vspace{0.6em}
{\normalsize\@author\par}
\makeatother
\endgroup
\vspace{0.4em}
\begin{abstract}
We prove that, for every fixed rational $c>0$, computing a $2^{-n^c}$-approximate Nash equilibrium in an $n$-player two-strategy anonymous game is $\PPAD$-complete. This holds even when one strategy always pays $1/2$ and, for all but two players, the other is constant outside at most five opponent counts. The reduction uses coefficient estimates to simulate circuit constraints approximately. Exact computation requires eliminating the simulation error: with three strategies, a separate construction uses polynomial roots to implement circuit gates exactly. This yields $\FIXP$-completeness of exact equilibrium computation under separable-linear reductions and $\FIXPa$-completeness of approximation in distance to an exact equilibrium.
\end{abstract}
\begingroup
\small
\setlength{\parskip}{0pt}
\makeatletter
\renewcommand{\l@section}[2]{%
 \addpenalty{\@secpenalty}\addvspace{0.2em}%
 \@dottedtocline{1}{0em}{1.8em}{\bfseries#1}{\bfseries#2}}
\renewcommand{\tableofcontents}{%
 \par\medskip\noindent{\normalsize\bfseries\contentsname}\par\smallskip
 \@starttoc{toc}}
\makeatother
\tableofcontents
\endgroup
\end{titlepage}
\pagenumbering{arabic}

\section{Introduction}\label{sec:introduction}

Many strategic decisions depend on aggregate demand. When agents choose
among shared resources, congestion depends on how many agents use each
resource, while preferences may differ across agents~\cite{mil96}. Nash
equilibrium describes an outcome in which no agent can improve her payoff
by deviating alone. Computing such outcomes is central to assessing what
these models can predict in large populations.

\emph{Anonymous games} isolate this form of interaction: a player's payoff
depends on her own action and on the number of opponents choosing each
action, while payoff functions may differ across players~\cite{blo00,dp15}.
For $n$ players and a fixed number of actions, the payoff tables have
polynomial size, despite exponentially many pure profiles. The model
therefore offers a compact description of heterogeneous incentives.
Does replacing individual interactions by counts also make equilibria
easier to compute?

The approximation algorithms suggest that aggregation is powerful.
For every fixed number of actions and every fixed $\varepsilon>0$,
an $\varepsilon$-Nash equilibrium can be found in polynomial
time~\cite{dp07,dp08,dp15}. The guarantee can even reach
$O(n^{-1+\alpha})$ for every fixed $\alpha>0$~\cite{cds17}.
An improvement to $O(n^{-1-\alpha})$ for any fixed $\alpha>0$ would
yield a fully polynomial-time approximation scheme~\cite{cds17}.
At finer precision, however, the distribution of action counts can encode
information about individual mixed strategies. Chen, Durfee, and
Orfanou~\cite{cdo15} used this information to prove $\PPAD$-completeness
with seven actions; Cheng, Diakonikolas, and Stewart~\cite{cds17}
observed that the construction can be reduced to five.
The remaining question is how much room a mixed strategy needs to encode
a hard computation.

Two actions leave each player with a single probability. That probability
must both represent a variable and make its contribution identifiable
inside an aggregate distribution. A third action supplies a second degree
of freedom: one probability can distinguish a player, while another
represents her value. Exact computation asks more of this encoding than
high-precision approximation. A small simulation error is tolerable in
the latter; the former requires identities that hold exactly. Rational
payoffs alone do not remove this distinction, since even two-strategy
anonymous games can have only irrational equilibria~\cite{gt17}.

We show that two actions already suffice for high-precision $\PPAD$
hardness, and that three suffice for exact $\FIXP$ hardness. The reductions
explain how anonymous counts recover individual interactions: selected
coefficients approximately isolate quadratic terms with two actions,
while polynomial roots recover individual probabilities exactly with
three. In both constructions, the required equilibrium supports follow
from the payoff inequalities themselves. This makes the recovery valid
for every target equilibrium and prevents boundary supports from escaping
the reduction.

\subsection{Our Results}

All games in our reductions have explicit rational payoff tables with entries in $[0,1]$. An $\varepsilon$-approximate Nash equilibrium allows no unilateral deviation to improve expected payoff by more than $\varepsilon$. Our first result establishes hardness with two actions even under a restricted payoff structure.

\begin{theorem}[PPAD-Completeness with Two Strategies]\label{thm:two-strategy-main}
For every fixed rational $c>0$, computing a $2^{-n^c}$-approximate Nash equilibrium of an $n$-player two-strategy anonymous game is $\PPAD$-complete. This holds even when strategy $0$ always pays $1/2$ and, for all but at most two players, the payoff of strategy $1$ is constant outside at most five opponent counts.
\end{theorem}

A third action permits exact algebraic simulation.

\begin{theorem}[FIXP-Completeness with Three Strategies]\label{thm:three-strategy-main}
Computing an exact Nash equilibrium of a three-strategy anonymous game is $\FIXP$-complete under separable-linear reductions.
\end{theorem}

A separable-linear reduction recovers each source coordinate from a single equilibrium probability by a rational affine map. Our exact reduction has this property, so it also preserves approximation in distance to an exact solution.

\begin{corollary}[Approximation in Distance]\label{cor:distance-main}
Given a positive rational $\varepsilon$ encoded in binary, computing a rational vector within $\varepsilon$ in the $\ell_\infty$ norm of some exact equilibrium of a three-strategy anonymous game is $\FIXPa$-complete.
\end{corollary}

The two hardness results concern different precision requirements. \Cref{thm:two-strategy-main} establishes high-precision hardness with the minimum nontrivial number of pure strategies. \Cref{thm:three-strategy-main} establishes exact hardness with three strategies. The latter does not follow by setting the approximation error to zero in the former: the two-strategy construction has a small coefficient-approximation error even at an exact equilibrium. The exact complexity with two strategies remains open. \Cref{tab:complexity-landscape} places the results alongside the known approximation algorithms.

\begin{table}[t]
\centering
\renewcommand{\arraystretch}{1.18}
\begin{tabularx}{\linewidth}{@{}l l X@{}}
\toprule
Strategies & Required Accuracy & Complexity \\
\midrule
Fixed $k$ & Fixed $\varepsilon>0$ & Polynomial time~\cite{dp07,dp08,dp15} \\
Fixed $k$ & $O(n^{-1+\alpha})$, fixed $\alpha>0$ & Polynomial time~\cite{cds17} \\
Two & $2^{-n^c}$, fixed $c>0$ & $\PPAD$-complete (\Cref{thm:two-strategy-main}) \\
Three & Exact equilibrium & $\FIXP$-complete (\Cref{thm:three-strategy-main}) \\
Three & Distance to an exact equilibrium & $\FIXPa$-complete (\Cref{cor:distance-main}) \\
Two & Exact equilibrium & In $\FIXP$; hardness remains open \\
\bottomrule
\end{tabularx}
\caption{Equilibrium Computation in Anonymous Games. The additive payoff guarantee and distance to an exact equilibrium are different approximation requirements.}
\label{tab:complexity-landscape}
\end{table}

\paragraph{Proof Techniques.}

Both reductions encode circuit variables in mixed strategies and use
best-response conditions to enforce the circuit constraints. The obstacle
is anonymity: a gate must depend on its designated inputs, whereas a
player's payoff can depend only on how many opponents choose each action.
We overcome this restriction in two ways.

With two strategies, we first reduce \PureCircuit{} to a sparse system
of quadratic equations. Each variable is represented by a pair of
players. Different pairs use well-separated probability scales, so
selected count probabilities approximately isolate the quadratic terms
needed for an equation. Changing a variable also affects the count
probabilities used by later equations. We add control players that keep
the leading contribution of each completed variable block
approximately constant. The main estimate bounds the quadratic terms
before any variable player is known to mix. It lets us prove that each
variable pair mixes, so approximate indifference enforces the intended
quadratic equations within the required error.

With three strategies $s,c,t$, we fix each player's total probability of
choosing $s$ or $c$, and encode a variable in how this probability is
split between the two actions. We choose distinct totals for different
players. The generating polynomial of the count distribution has one
factor per player; evaluating suitable derivatives at the roots of these
factors isolates individual encoded values. This gives anonymous payoff
tables whose expectations equal a designated variable. An additional
gadget computes products exactly, allowing the best-response conditions
to enforce the algebraic circuit equations without approximation error.

The correctness proofs must justify every use of indifference. We first
show from the best-response inequalities that the required actions have
positive probability. Only then do we use equal, or approximately equal,
expected payoffs to derive the circuit constraints. This ensures that
every equilibrium at the required accuracy yields a circuit solution.
\Cref{sec:recovery-principles} explains the constructions and these
support arguments in more detail.

\subsection{Related Work}
\paragraph{Approximation Algorithms.}
Approximation schemes for fixed strategy counts were developed by
Daskalakis and Papadimitriou~\cite{dp07,dp08,dp15} and, for two
strategies, Daskalakis~\cite{das08}. The relevant distributional tools
include moment matching~\cite{dp09,dp15} and structural descriptions of
sums of independent random variables~\cite{dkt15,ddkt16}. Fourier methods
give polynomial dependence on the population and quasipolynomial
dependence on $1/\varepsilon$ for every fixed number of
strategies~\cite{dks16a}; for two strategies, the latter dependence can be
improved to $(1/\varepsilon)^{O(\log(1/\varepsilon))}$~\cite{dks16b}.
Our hardness result concerns exponentially small $\varepsilon$, where
these approximation schemes need not run in polynomial time.

\paragraph{Hardness of Nash Equilibria.}
The classical $\PPAD$-completeness results concern normal-form
games~\cite{dgp09,cdt09}. Constant-error hardness also holds for two-action
graphical polymatrix games of degree three~\cite{rub18}, where payoffs
may depend on specified opponents. Chen, Durfee, and Orfanou~\cite{cdo15}
overcame the restriction to anonymous payoffs using their generalized
radix game. It uses widely separated probabilities to distinguish
players, with additional strategies to encode values and enforce these
probabilities. Our two-strategy result removes the need for additional
strategies.

\paragraph{Exact Computation.}
Etessami and Yannakakis~\cite{ey} introduced $\FIXP$ and proved exact
equilibrium computation $\FIXP$-complete for three-player games. Our exact
reduction uses their framework and the restricted \ExactGCircuit{}
problem of Filos-Ratsikas, Giannakopoulos, Hollender, Lazos, and
Po\c{c}as~\cite{FGHLP23}. Garg, Mehta, Vazirani, and
Yazdanbod~\cite{gmvy18} proved $\FIXPa$-completeness for symmetric
equilibria in symmetric three-player games. Their number of strategies
grows with the input. Our reduction fixes the number of strategies at
three and decodes every equilibrium by a separable-linear map.

\paragraph{Organization.}
\Cref{sec:preliminaries} defines the equilibrium notions and circuit sources.
\Cref{sec:recovery-principles} gives an overview of the two reductions and
their quantitative guarantees. \Cref{sec:two-strategy-overview,sec:three-strategy-overview}
give the constructions and support arguments, and \Cref{sec:complexity}
derives the complexity classifications. The appendices supply the longer
coefficient, interpolation, and precision estimates.

\section{Preliminaries}\label{sec:preliminaries}

An $n$-player anonymous game has a common action set $\mathcal S$ of size $k$. It specifies $U_p(a,h)$ for every player $p$, action $a\in\mathcal S$, and opponent histogram $h\in\mathbb Z_{\ge0}^{\mathcal S}$ with $\sum_a h_a=n-1$. There are $k\binom{n+k-2}{k-1}$ entries per player. We use explicit rational payoff tables, and measure running time in their total encoding length. In a mixed profile $\sigma$, players randomize independently, and $\sigma_{p,a}$ is the probability that player $p$ chooses $a$. We write $u_p(a,\sigma_{-p})$ for a pure action's expected payoff and \[v_p(\sigma)=\sum_a\sigma_{p,a}u_p(a,\sigma_{-p})\] for the player's mixed payoff.

\begin{definition}[Approximate Nash Equilibria]\label{def:equilibrium}
  An \emph{$\varepsilon$-approximate Nash equilibrium} (ANE) is a mixed profile $\sigma$ in which no player can improve her expected payoff by more than $\varepsilon$ through a unilateral deviation: $u_p(a,\sigma_{-p})\le v_p(\sigma)+\varepsilon$ for every $p,a$.

  An \emph{$\varepsilon$-well-supported Nash equilibrium} (WSNE) requires each action played with positive probability to be within $\varepsilon$ of a best response: $u_p(a,\sigma_{-p})\ge\max_b u_p(b,\sigma_{-p})-\varepsilon$ whenever $\sigma_{p,a}>0$.

  We abbreviate these notions as $\varepsilon$-$\ANE$ and $\varepsilon$-$\WSNE$. At $\varepsilon=0$, both are exact Nash equilibria.
\end{definition}

For fixed $k$, expected payoffs are computable in polynomial time by convolving the opponents' categorical distributions~\cite{cdo15,pr08}. The same recurrence gives polynomial-size arithmetic circuits. We use well-supported equilibrium within the reductions: a supported action gives the corresponding payoff inequality, and two supported actions give a two-sided comparison.

For a finite vector $\omega=(\omega_1,\ldots,\omega_m)$, its $r$th
\emph{elementary symmetric polynomial} is
\[
 e_r(\omega)=\sum_{\substack{S\subseteq[m]\\|S|=r}}
                  \prod_{j\in S}\omega_j
 \qquad (0\le r\le m).
\]
The empty product gives $e_0(\omega)=1$, and we set $e_r(\omega)=0$
for $r<0$ or $r>m$. In particular, $e_1(\omega)=\sum_j\omega_j$
and $e_2(\omega)=\sum_{i<j}\omega_i\omega_j$.

For two strategies, let $p_j=\sigma_{j,1}<1$ and define its \emph{weight} $\omega_j=p_j/(1-p_j)$, the ratio of its action-$1$ and action-$0$ probabilities. Let $K_{-p}$ count opponents of $p$ choosing action $1$, and put $P_{0,p}=\prod_{j\ne p}(1-p_j)$. Independence gives
\begin{equation}\label{eq:two-strategy-generating}
 \Pr(K_{-p}=r)=P_{0,p}e_r(\omega_{-p}),\qquad
 \prod_j(1+\omega_jz)=\sum_r e_r(\omega)z^r.
\end{equation}
Here $\omega_{-p}$ is the vector with player $p$'s coordinate deleted;
a superscript $-p$ likewise denotes deletion of that player.
Thus an anonymous payoff is a linear form of count coefficients.
Our two-strategy argument proves $p_j<1$ before introducing these weights.

\subsection{Complexity and Approximation}\label{sec:complexity-approximation}

$\PPAD$ is the parity-based class of total search problems~\cite{pap94}. $\FIXP$ concerns exact fixed points of algebraic circuits, and $\FIXPa$ concerns rational approximation in distance to an exact fixed point~\cite{ey}. Approximation problems ask for rational outputs; exact computation is understood as a real-valued search problem in this framework. The payoff guarantee of an approximate equilibrium need not control its distance to an exact equilibrium. For exact hardness we use separable-linear reductions, in which source coordinate $i$ is recovered as $\alpha_i\sigma_{p_i,a_i}+\beta_i$ with rational coefficients of polynomial encoding length. We measure distance in the $\ell_\infty$ norm on the full vector of action probabilities. For an interval $I=[a,b]$, $\clip_I(z)=\min\{b,\max\{a,z\}\}$ denotes projection to $I$; this map is $1$-Lipschitz.

\paragraph{Approximation and Normalization.}

A common transformation $U\mapsto aU+b$, $a>0$, preserves exact equilibria and sends an $\varepsilon$-approximate equilibrium to an $a\varepsilon$-approximate equilibrium. For each fixed number of actions, computing a $2^{-\poly(n)}$-approximate anonymous equilibrium belongs to $\PPAD$~\cite{cdo15}; the fixed-point formulation uses the weak approximation results of Etessami and Yannakakis~\cite{ey}. Exact computation belongs to $\FIXP$ through the usual Nash fixed-point map~\cite{nash51,ey}.

Our payoff arguments use supported actions, so we need the following standard conversion from approximate to well-supported Nash equilibria.

\begin{lemma}[Approximation Conversion, {\cite{cdo15}}]\label{lem:ane-wsne}
In an $n$-player game with $k$ actions and payoffs in $[0,1]$, a $\delta$-$\ANE$ can be converted in polynomial time into an $\varepsilon$-$\WSNE$ if $\delta\le\varepsilon^2/(16kn)$ and pure-action expected payoffs are polynomial-time computable.
\end{lemma}

The conversion removes probability from actions whose payoff is far below a best response. Its square-root loss in accuracy is absorbed in our precision choices. We use the bound $\varepsilon^2/(32n)$ for two actions and $\varepsilon^2/(48n)$ for three.

\subsection{Circuit Problems}\label{sec:circuit-sources}

The two-strategy reduction starts from a discrete circuit problem,
whereas the three-strategy reduction starts from an algebraic one.
Both allow cycles and ask for assignments that satisfy every gate
simultaneously.

\paragraph{\PureCircuit.}

Our discrete hardness source is \PureCircuit. Its gates permit an undetermined value, while purification ensures that some Boolean information survives.

\begin{definition}[\PureCircuit, {\cite{pure}}]\label{def:pure-circuit}
An instance consists of wires and gates. Each gate involves three distinct wires, and each wire is an output of exactly one gate; cycles are permitted. A solution assigns each wire a value in $\{0,1,\bot\}$, where $\bot$ means undetermined, subject to these conditions:
\begin{itemize}[nosep]
\item A $\NOR$ gate has two inputs and one output. Two zero inputs force output $1$; either input equal to $1$ forces output $0$. Otherwise the output is unrestricted.
\item A $\PURIFY$ gate has one input and two outputs. At least one output is Boolean. A Boolean input forces both outputs to equal it.
\end{itemize}
\end{definition}

\begin{theorem}[\PureCircuit{} Completeness, {\cite{pure}}]\label{thm:pure-circuit}
Finding a solution to a \PureCircuit{} instance with $\NOR$ and $\PURIFY$ gates is $\PPAD$-complete.
\end{theorem}

\paragraph{\ExactGCircuit.}

Our exact hardness source is \ExactGCircuit{} over a basis of three gate
types~\cite{FGHLP23}. An instance is a generalized circuit whose gates define
a continuous map of the unit cube. A solution satisfies all defining
equations, including those on cycles.

\begin{definition}[Restricted Generalized Circuits]\label{def:restricted-circuit}
A restricted generalized circuit has variables $x_1,\ldots,x_G\in[0,1]$ and exactly one
defining gate for each output $x_i$. Each gate is of one of the
following types, with input indices in $[G]$:
\begin{itemize}[nosep]
\item $G_{1-}$: $x_i=1-x_j$.
\item $G_{\times2}$: $x_i=\min\{2x_j,1\}$.
\item $G_\phi$: $x_i=(1+x_j)(1+x_k)/4$.
\end{itemize}
Dependencies may contain cycles. For $\varepsilon\ge0$, an
\emph{$\varepsilon$-solution} is a vector $x\in[0,1]^G$ for which
the absolute difference between the two sides of each gate equation
is at most $\varepsilon$. A $0$-solution is an \emph{exact solution}.
The search problem \ExactGCircuit{} asks for an exact solution;
for $\varepsilon>0$, $\varepsilon$-\GCircuit{} asks for a rational
$\varepsilon$-solution. We use both problems with the gate basis above.
\end{definition}

Each gate maps its inputs into $[0,1]$, so an exact solution exists
by Brouwer's Theorem. The restricted basis already captures the
complexity of algebraic fixed points.

\begin{theorem}[Generalized Circuit Completeness, {\cite{FGHLP23}}]\label{t:thm:source}
With the gate basis in \Cref{def:restricted-circuit}, \ExactGCircuit{}
is $\FIXP$-complete under separable-linear reductions.
For some fixed rational $\varepsilon_{\mathrm{GC}}>0$,
$\varepsilon_{\mathrm{GC}}$-\GCircuit{} is $\PPAD$-complete.
Hardness holds even when every gate has distinct input and output
variables.
\end{theorem}

Our three-strategy reduction uses the exact part of this theorem.
Its construction also accommodates repeated inputs and self-loops.

For two strategies, we first express \PureCircuit{} as a sparse quadratic system tailored to count coefficients; \Cref{b:sec:source} gives this intermediate problem and the reduction.

\section{Overview of the Reductions}\label{sec:recovery-principles}

We outline the reductions proving $\PPAD$-hardness with two strategies
and $\FIXP$-hardness with three. Both simulate circuit constraints using
anonymous payoff tables. The first uses coefficient estimates and allows
a small simulation error; the second uses polynomial roots to obtain
exact identities. In both reductions, the proof must first establish
the equilibrium supports that make these simulations valid.

\subsection{The PPAD-Hardness Reduction}\label{b:sec:source}

The source problem asks for a point in $[0,1]^q$ that approximately
satisfies a sparse system of quadratic equations. We use this formulation
because each input term $x_j(1-x_j)$ can be represented by the product
of two players' weights. The coefficient margins will later let us
prove that both players mix.

Fix $\delta=2^{-7}$ and $\eps_{\mathrm{src}}=2^{-17}$.

\begin{definition}[\OQE]\label{b:def:source}
\emph{Input.}
A positive integer $q$, a type $\mathrm L$ or $\mathrm Q$ for each
variable $x_i$, $i\in[q]$, and rational coefficients defining the following
system on $x_1,\ldots,x_q\in[0,1]$.
Write $y_i=x_i(1-x_i)$. Variable $x_i$ satisfies
\begin{align}
 x_i&=1-a_i-\sum_{j<i}b_{ij}y_j
                &&\text{(type L)},\label{b:eq:linear}\\*
 x_i^2-(1-a_i)x_i&=c_i+\sum_{j>i}d_{ij}y_j
                &&\text{(type Q)}.\label{b:eq:quadratic}
\end{align}
All coefficients are rational and nonnegative, and at most two input
coefficients in any equation are nonzero. Type L equations satisfy
\begin{equation}\label{b:eq:linear-margin}
 a_i=\frac{1+\delta}{2},\qquad
 a_i+\frac14\sum_j b_{ij}\le1-\frac{\delta}{2}.
\end{equation}
Type Q equations satisfy
\begin{equation}\label{b:eq:quadratic-margin}
 a_i=\frac{9}{10},\qquad c_i>\frac{1}{5},\qquad
 c_i+\frac14\sum_jd_{ij}\le a_i-\delta.
\end{equation}
\emph{Output.}
A rational vector $x\in[0,1]^q$ for which the left- and right-hand
sides of every equation differ by at most $\eps_{\mathrm{src}}$.
We call such a vector an \emph{$\eps_{\mathrm{src}}$-solution}.
\end{definition}

The ordering restricts the direction of each equation type:
type L uses earlier variables, whereas type Q uses later variables.
It is not a topological order; the combined dependency graph may
contain cycles.
The coefficient margins also imply the bounds
$\sum_jb_{ij}<2$ and $c_i+\sum_jd_{ij}<4$ used below.

These equations define a continuous map of $[0,1]^q$ to itself:
a type L output lies in $[\delta/2,(1-\delta)/2]$, and the positive
type Q root lies in $(1/2,1)$ and is the unique root in $[0,1]$.
Brouwer's Theorem gives an exact solution. The bounded coefficients
then permit rounding to a rational $\varepsilon_{\mathrm{src}}$-solution
of polynomial encoding length, so \OQE{} is total.

The following lemma reduces \PureCircuit{} to this quadratic form.

\begin{lemma}\label{b:lem:source}
\OQE{} is $\PPAD$-hard, even when all coefficients belong to a fixed
finite set of rational numbers.
\end{lemma}

\begin{proof}[Proof Outline]
We reduce from the $\PPAD$-complete problem in \Cref{thm:pure-circuit}.
Represent a wire by $z=4x(1-x)$, decoding $z\le1/12$ as $0$,
$z\ge11/12$ as $1$, and the intervening values as $\bot$.
A type Q equation followed by a constant-length chain of type L
equations implements each gate, with a Boolean margin exceeding
the accumulated error at tolerance $\eps_{\mathrm{src}}$.
Place all type Q variables first and then each type L chain in its
forward order. Type Q equations read the later wire outputs, while
type L equations read earlier chain entries, so circuit cycles are
preserved. The construction uses at most twelve variables per wire
and fixed rational coefficients. \Cref{b:app:source} gives the gate
formulas, verifies every coefficient margin, and proves the
approximate-solution decoding.
\end{proof}

The anonymous game uses four players per source variable and three
additional players. Its accuracy is stated explicitly here; population
padding in \Cref{sec:complexity} gives \Cref{thm:two-strategy-main}.

\begin{theorem}[Quantitative Two-Strategy Reduction]\label{b:thm:core}
There is a polynomial-time construction mapping a $q$-variable instance of \OQE{} to an $N=4q+3$ player two-strategy anonymous game. From every $N^{-54N^2}$-$\ANE$, an $\varepsilon_{\mathrm{src}}$-solution can be computed in polynomial time, where $\varepsilon_{\mathrm{src}}=2^{-17}$. All payoffs are rational and in $[0,1]$; strategy $0$ always pays $1/2$, and every encoding player's strategy-$1$ payoff is constant outside at most five counts. There are only two other players.
\end{theorem}

\paragraph{The Coefficient Bound Needed before Mixing.}
For two nonnegative weights with total weight $H$, their product is
at most $H^2/4$. Anonymous coefficients also include products involving
other players, so this elementary observation alone cannot justify the
same bound. The third coefficient measures the extra contribution.
We write $s_i$ for the scale of variable pair $i$ and
$\tau_i=t^2s_i$ for its subsequent control scale, with
a separation parameter $t>0$ chosen so that $Nt$ is small.

\begin{lemma}[Second Coefficients with Small Third Coefficients]\label{b:lem:quarter-bound}
For every finite nonnegative list $J$, with $H=e_1(J)$,
\begin{equation}\label{b:eq:third-order}
 4e_2(J)^2\le H^2e_2(J)+3He_3(J).
\end{equation}
If $e_3(J)\le\tau e_2(J)$ for $\tau\ge0$, then
$e_2(J)\le H^2/4+(3/4)\tau H$.
The coefficient $3/4$ is attained by three equal positive weights
when $\tau=e_3(J)/e_2(J)$.
\end{lemma}

\begin{proof}
Write $p_r=\sum_{a\in J}a^r$. Expansion gives
$H^2e_2+3He_3-4e_2^2=Hp_3-p_2^2\ge0$
by Cauchy--Schwarz.
If $e_2>0$, divide by $4e_2$ and use $e_3\le\tau e_2$;
otherwise the second bound is immediate.
Three equal weights give equality.
\end{proof}

The contribution of the reduction is to enforce these hypotheses on
arbitrary equilibrium supports. If a normalized prefix coefficient is
small, the quadratic bound is immediate. Otherwise, we factor out the
larger weights, and the control payoffs give
$H\le s_i(1+O(Nt))$ and $e_3(J)\le\tau_i(1+O(Nt))e_2(J)$,
with $\tau_i/s_i=t^2$.
\Cref{b:lem:quarter-bound} then bounds the normalized quadratic
coefficient by $1/4+O(Nt)$.
The source margins use exactly this bound to force the encoding pair
to mix; a generic bound of $1/2$ would leave no useful gap.

\paragraph{Why Later Coefficients Need Compensation.}
Write the pair's weights as $s_i u_i,s_i v_i$.
The control equations will give $u_i+v_i=1+O(Nt)$, so
$x_i=u_i/(u_i+v_i)$ and $u_iv_i=x_i(1-x_i)+O(Nt)$.
Pairs with $(u_i,v_i)=(1/2,1/2)$ and $(1/10,9/10)$ have the same
sum but products $1/4$ and $9/100$. Thus a later coefficient that
selects both players inherits an unwanted dependence on $x_i$.
We assign a later control player a weight approximately
$\chi_{i+1}/(u_iv_i)$ to cancel that factor. Here $\chi_i,s_i,\tau_i$
are successively separated reference scales, specified in
\Cref{b:sec:construction}.

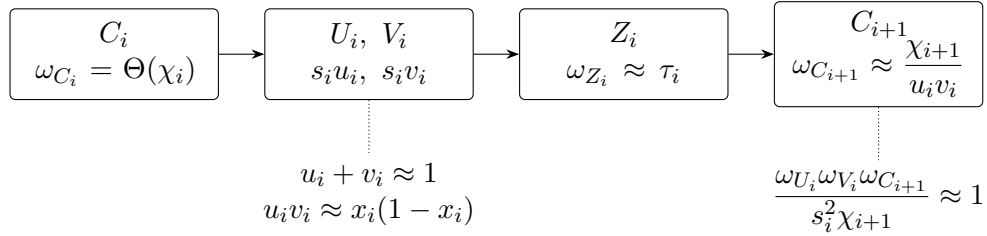
\begin{figure}[htbp]
\centering
\begin{tikzpicture}[
 node distance=6mm,
 box/.style={draw,rounded corners=2pt,align=center,
             minimum height=12mm,text width=25mm},
 >=Stealth]
\node[box] (c) {$C_i$\\$\omega_{C_i}=\Theta(\chi_i)$};
\node[box,right=of c] (pair) {$U_i,\ V_i$\\$s_i u_i,\ s_i v_i$};
\node[box,right=of pair] (z) {$Z_i$\\$\omega_{Z_i}\approx\tau_i$};
\node[box,right=of z] (next) {$C_{i+1}$\\$\displaystyle
 \omega_{C_{i+1}}\approx\frac{\chi_{i+1}}{u_iv_i}$};
\draw[->] (c) -- (pair);
\draw[->] (pair) -- (z);
\draw[->] (z) -- (next);
\node[align=center,below=7mm of pair] (sum)
 {$u_i+v_i\approx1$\\$u_iv_i\approx x_i(1-x_i)$};
\node[align=center,below=7mm of next] (prod)
 {$\displaystyle
 \frac{\omega_{U_i}\omega_{V_i}\omega_{C_{i+1}}}{s_i^2\chi_{i+1}}\approx1$};
\draw[densely dotted] (pair) -- (sum);
\draw[densely dotted] (next) -- (prod);
\end{tikzpicture}
\caption{Players and Constraints for One Variable.
The product constraint cancels $u_iv_i$ from the next normalized prefix.
All approximate relations follow from the equilibrium inequalities in
\cref{b:lem:normalization}.}
\label{b:fig:block}
\end{figure}

\paragraph{Correctness of the Reduction.}
The proof proceeds in three steps. First, two auxiliary players bound
the total probability of action $1$, making all weights finite and
the active encoding players an initial segment.
Second, the control rows bound the normalized count coefficients on
every support, as explained above. Third, a forward induction forces
each pair and its controls to mix and transfers the compensated
prefix to the next block. The resulting indifference equations recover
the source equations with residual $O(Nt)$.

The order of the source equations resolves a further deletion effect.
Deleting $U_i$ leaves an earlier prefix intact, but divides a later
prefix product by $s_i u_i$. Type L reads only earlier terms; type Q
reads later terms and multiplies its payoff equation by $u_i$ to
cancel this factor. Each input term uses one count, giving the
five-count restriction in \Cref{b:thm:core}. The support lemmas in
\Cref{b:sec:correctness} implement this argument before any source
coordinate is decoded.

\subsection{The FIXP-Hardness Reduction}\label{sec:three-strategy-principle}

A third action separates identity from value. For actions $s,c,t$,
let $w_P=\sigma_{P,s}+\sigma_{P,c}$ be player $P$'s signal probability
and $\theta_P=\sigma_{P,c}/w_P$ her conditional value, defined when $w_P>0$.
We prescribe distinct rational signal probabilities $\delta_P\in(0,1)$
and encode $x_i$ as $u_i=1/4+x_i/2$ in a designated player $X_i$.
This interior encoding will let us establish support without assuming
that the source equations already hold.

The quantitative guarantee makes the exact reduction explicit.

\begin{theorem}[Quantitative Three-Strategy Reduction]\label{t:thm:quantitative-reduction}
From a restricted generalized circuit $\mathcal C$ with $G$ gates, one can construct in polynomial time a three-strategy anonymous game $\Gamma(\mathcal C)$ with $N=O(G\log(2G))$ players and rational payoffs in $[0,1]$. There are rational parameters of polynomial encoding length $\xi_{\mathrm{out}}\ge2^{-O(N)}$ and $\Lambda_{\mathrm{out}}\le2^{O(N)}$ such that, for $0\le\varepsilon\le\xi_{\mathrm{out}}$, every $\varepsilon$-$\WSNE$ decodes to a $\Lambda_{\mathrm{out}}\varepsilon$-solution. At zero error the decoder is separable-linear:
\begin{equation}\label{t:eq:main-SL-decoder}
 x_i=\frac{2\sigma_{X_i,c}}{\delta_{X_i}}-\frac12.
\end{equation}
Every logical player has $\delta_{X_i}\in[1/5,4/5]$, so this decoder is
$10$-Lipschitz in the $\ell_\infty$ norm.
\end{theorem}

\paragraph{A Root Eliminates the Other Players.}
At $w=\delta$, let $A_j(z)=1-\delta_j+\delta_jz$ and
$\zeta_j=1-1/\delta_j$. With two opponents $j,k$, the histogram
statistics $K=h_s+h_c$ and $L=h_c$ satisfy
$\E[Lz^{K-1}]=\delta_j\theta_jA_k(z)+\delta_k\theta_kA_j(z)$.
At $z=\zeta_j$, the second summand vanishes, and division by the
nonzero rational $\delta_jA_k(\zeta_j)$ recovers $\theta_j$.
The same cancellation holds with any number of opponents.
Second-order statistics leave products involving the player at the
chosen root. The following construction cancels the unwanted terms.

\begin{figure}[htbp]
\centering
\begin{tikzpicture}[>=Stealth,
 box/.style={draw,rounded corners=2pt,align=center,minimum height=13mm},
 every node/.style={font=\small}]
\node[box,text width=28mm] (counts) at (0,0)
 {Opponent\\Histogram\\$K=h_s+h_c$\\$L=h_c$};
\node[box,text width=44mm] (first) at (5.1,1.15)
 {First-Order Statistic\\$\E[Lz^{K-1}]$};
\node[box,text width=44mm] (second) at (5.1,-1.15)
 {Second-Order Statistic\\$\E[\binom L2 z^{K-2}]$};
\node[box,text width=34mm] (value) at (11.2,1.15)
 {Conditional Value\\$\theta_j$};
\node[box,text width=34mm] (product) at (11.2,-1.15)
 {A Product\\$\theta_j\theta_k$};
\draw[->] (counts) -- (first);
\draw[->] (counts) -- (second);
\draw[->] (first) -- node[above,align=center] {$z=\zeta_j$\\Normalize} (value);
\draw[->] (second) -- node[below] {$z=\zeta_B$} (product);
\end{tikzpicture}
\caption{Individual Probabilities from Anonymous Counts. A root makes
all unwanted first-order terms vanish. At second order, auxiliary equilibrium
equations cancel the remaining pair products.}
\label{fig:root-recovery}
\end{figure}
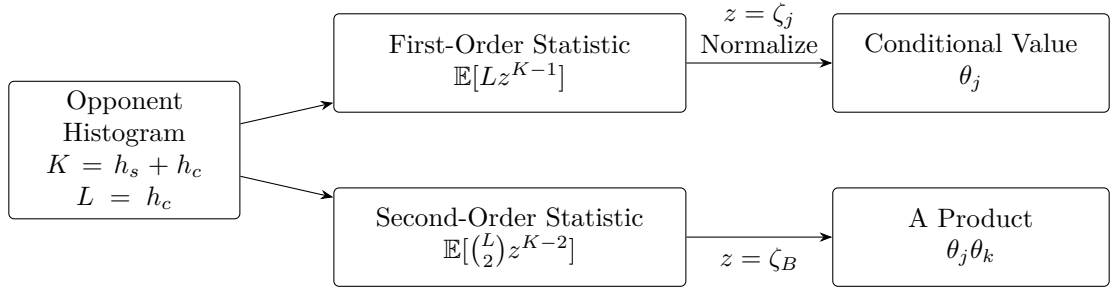

\paragraph{Exact Multiplication by Cancellation.}
To see the algebra, work at the prescribed signals and set the filler
players' conditional values to zero. These players will separate active
nodes. Consider inputs $u_j,u_k$ and four auxiliary players $B,C,V,W$ with
conditional values $\beta,\chi,\nu,\psi$, respectively; we suppress
the gate index in this calculation. Place $V,B$ at adjacent roots
and put $h=\zeta_B-\zeta_V>0$ and $d=1/h$.
Fix an observer distinct from the four auxiliaries and the input players.
Let $S$ be the sum of $\theta_p/(\zeta_B-\zeta_p)$ over its active opponents
other than $B,V,W$. Subtracting a first-order term from the second-order
statistic gives the expected value
\[
 \beta\left[d\left(\nu-\frac12\right)
       +\frac{\psi}{\zeta_B-\zeta_W}+S\right].
\]
Define
\[
 F^{\mathrm{mul}}
 =\frac12+h\left(u_k-\frac{1}{2(\zeta_B-\zeta_W)}-S\right).
\]
The four auxiliary payoff differences have ideal residuals
$D_B=\chi-1/2$, $D_C=u_j-\beta$, $D_V=\psi-1/2$, and
$D_W=F^{\mathrm{mul}}-\nu$.
Each residual depends only on opponents' conditional values. In particular, $W$ need not read
its own value $\psi$: its target uses the constant $1/2$ that the
$V,W$ equations enforce. When the four residuals vanish, $dh=1$ gives
\begin{align*}
 d\left(\nu-\frac12\right)
   &=u_k-\frac{1}{2(\zeta_B-\zeta_W)}-S,\\
 \beta\left[d\left(\nu-\frac12\right)
       +\frac{\psi}{\zeta_B-\zeta_W}+S\right]
   &=\beta u_k=u_j u_k.
\end{align*}
Thus the multiplication is exact, with no truncation of smaller terms.

The support argument must justify setting these residuals to zero.
We separate active players by filler players, whose conditional values
are forced to zero, except for the adjacent pair $V,B$.
This spacing keeps $F^{\mathrm{mul}}$ in $[15/32,17/32]$ for every
choice of the other conditional values. The $V,W$ inequalities then
force both values into the interior before any source equation is
used. At zero gate error, cancellation first produces $\beta u_k$; the projection
constraints next put every $u_j$ in the interior, and only then do
the $B,C$ inequalities force $\beta=u_j$.
\Cref{t:sec:semantic-framework,t:sec:gate-overview} give the anonymous
tables and the quantitative support proof. The calculation above is
at the prescribed signal probabilities; we next explain how to
enforce them at every equilibrium.

\paragraph{Calibration Must Hold Globally.}
Root recovery requires prescribed signal probabilities.
Local invertibility alone cannot exclude a distant or boundary
equilibrium. We instead enforce a product map and use the following
global statement. Fekete points maximize the Vandermonde product
$\prod_{i<j}(y_j-y_i)$ over ordered $y\in[0,1]^N$.
Rational approximations to these points satisfy the product-gap
hypothesis below; \Cref{t:sec:fekete} constructs them.

For a signal vector $w$, write $c_i=\prod_{j\ne i}(\delta_j-\delta_i)$,
$c_*=\min_i|c_i|$, $b=\min_i\{\delta_i,1-\delta_i\}$, and
$g=\min_{i<N}(\delta_{i+1}-\delta_i)$. Define
\begin{equation}\label{t:eq:ratio-map}
 Q_i(w)=\prod_{j\ne i}(w_j-\delta_i),
 \qquad
 \varrho_i(w)=Q_i(w)/c_i.
\end{equation}
\begin{theorem}[Global Calibration on Arbitrary Supports]\label{t:thm:robust-localization}
Let $N$ be even, $0<\delta_1<\cdots<\delta_N<1$, and
$0<\rho_0\le g^2/(4N^3)$. Suppose the Lagrange polynomials
$\ell_i^\delta(z)=\prod_{j\ne i}(z-\delta_j)/(\delta_i-\delta_j)$ satisfy
\[
 \left|\prod_i\ell_i^\delta(w_i)\right|
 \le \exp\!\left(-\frac{\rho_0^2}{32N}\right)
 \quad\text{whenever }w\in[0,1]^N
 \text{ and }\|w-\delta\|_\infty\ge\rho_0.
\]
Consider a three-action game whose expected $s$-payoff is
$|c_i|(1-\varrho_i(w))$, whose $t$-payoff is zero, and whose
$c$-payoff differs from its $s$-payoff by at most $\gamma$.
Every $\xi$-$\WSNE$ with
\begin{equation}\label{t:eq:global-gamma}
 \gamma+\xi\le\frac{c_*b\rho_0^2}{512N^3}
\end{equation}
satisfies $\|w-\delta\|_\infty<\rho_0$ and
$\|w-\delta\|_\infty\le4N^2(\gamma+\xi)/(c_*b)$.
\end{theorem}

\emph{Proof Idea.}
Lagrange interpolation gives the conservation identity
$\sum_i(w_i-\delta_i)(\varrho_i(w)-1)=0$.
The equilibrium inequalities bound each summand from above,
even when $w_i=0$ or $1$; conservation supplies its lower bound.
Thus every ratio differs from one by at most
$2N(\gamma+\xi)/(c_*b)$.
Their product equals $\prod_i\ell_i^\delta(w_i)$.
The assumed gap excludes every point outside the $\rho_0$-box,
and local inversion gives the stated displacement bound.
\Cref{t:sec:calibration} proves each step, including the boundary cases.

The local step has a short proof that does not lose a factor
depending on the smallest node gap.

\begin{lemma}[Inverse Bound for Reciprocal Differences]\label{t:lem:elementary-inverse}
For even $N$ and distinct $a_1<\cdots<a_N$, the matrix $J_{ii}=0$, $J_{ij}=1/(a_j-a_i)$ satisfies $\|J^{-1}\|_2\le a_N-a_1$. For signal probabilities in $[0,1]$, $\|J^{-1}\|_\infty\le\sqrt N$.
\end{lemma}
\begin{proof}
With $D=\operatorname{diag}(a_i)$, $DJ-JD=I-\one\one^{\mathsf T}$. The least singular value $\sigma$ of a real skew-symmetric even-dimensional matrix has a right singular subspace of dimension at least two, even if $\sigma=0$. Choose a unit vector $v$ in that subspace with $\one^{\mathsf T}v=0$. For $m=(a_1+a_N)/2$,
$1=v^{\mathsf T}(DJ-JD)v =2v^{\mathsf T}(D-mI)Jv \le(a_N-a_1)\sigma$.
Hence $\sigma>0$ and the spectral bound follows. The infinity-norm bound uses $\|M\|_\infty\le\sqrt N\|M\|_2$.
\end{proof}

\paragraph{From Equilibria to Exact Circuit Solutions.}
Gate terms perturb only action $c$ by a small positive scale $\eta$.
Calibration first places every $w_P$ in the interior, and a first
application of the gate inequalities shows $\theta_P<1$.
Therefore both unperturbed actions $s,t$ are supported.
Comparing them removes the gate perturbation from the calibration
bound. A second application of the same gate inequalities now gives
errors proportional to the WSNE error alone.
At an exact equilibrium, we obtain $w=\delta$ and exact source
equations, making the affine decoder in \Cref{t:eq:main-SL-decoder}
valid. \Cref{t:sec:gate-overview} carries out these two applications.

\section{PPAD-Hardness with Two Strategies}\label{sec:two-strategy-overview}

We prove \Cref{b:thm:core} by simulating the quadratic system from
\Cref{b:sec:source} with anonymous payoffs. The players in
\Cref{b:fig:block} enforce the normalization and product compensation
described there. We specify their payoffs, establish the required supports,
and show how every approximate equilibrium yields a solution to the
quadratic system.

\subsection{Players and Payoffs}\label{b:sec:construction}

Given a source with $q$ variables, use $N=4q+3$ players.
The first $D=4q+1$ are called the \emph{encoding players}: player $1$,
followed by blocks $(C_i,U_i,V_i,Z_i)$ for $i=1,\ldots,q$.
The final two players are $A$ and $B$. Put $r_i=4i-2$, so the
indices of block $i$ are $r_i,r_i+1,r_i+2,r_i+3$.
Player $A$ has index $D+1$ and also serves as $C_{q+1}$;
player $B$ has index $N$.

The pair $(U_i,V_i)$ encodes variable $x_i$. The other players enforce
the conditions needed for this encoding: $C_i$ controls the total
weight assigned to the pair and later players, $V_i$ also controls the
weight after the pair, and $Z_i$ imposes the product constraint used
to normalize the next block. Players $A$ and $B$ bound the total
probability of strategy $1$. These roles guide the construction;
the correctness proof derives the required conditions from the
equilibrium inequalities.

Take $t=2^{-36}/N$. The proof will require only that the accumulated error
$Nt$ is a sufficiently small constant. We separate the variable-pair scale
from its preceding control scale by $t^4$, so coefficient factorization
remains accurate even before mixing is known. The other gaps control the
compensation of products. This choice is verified in \cref{b:sec:correctness}.
Every implicit constant in $O(\cdot)$ below is independent of $q,N,i,t$
and the source instance. We specify the fixed proof constants with
\cref{b:lem:block-step}.

Assign reference weights $\bar\omega_1=1/7$, $\chi_i=t^{7i-6}$,
$s_i=t^{7i-2}$, $\tau_i=t^{7i}$, and $\chi_{q+1}=t^{7q+1}$.
Here $\bar\omega_{C_i}=\chi_i$, $\bar\omega_{U_i}=\bar\omega_{V_i}=s_i$,
$\bar\omega_{Z_i}=\tau_i$, and $\bar\omega_A=\chi_{q+1}$.
These reference weights record the scales used in the construction.
The control weights vary with the encoded products to normalize
successive prefixes.
Put $S_0=1$ and $S_k=\prod_{\ell=1}^k\bar\omega_\ell$ for
$1\le k\le D+1$.
Set $K_i=s_i^2\chi_{i+1}$, $T=1/8$, $c_0=3/4$,
$b_0=c_0-T=5/8$, and $\beta=b_0/(1-b_0)=5/3$.
The scale relations used in the proof are
\begin{equation}\label{b:eq:scales}
 \frac{s_i}{\chi_i}=t^4,\qquad
 \frac{\tau_i}{s_i}=t^2,\qquad
 \frac{\chi_{i+1}}{\tau_i}=t,\qquad
 \frac{\tau_i^2}{s_i\chi_{i+1}}=t.
\end{equation}

\paragraph{Intended Constraints.}
We first explain the intended constraints, temporarily omitting the
small error terms. Write $\omega_{U_i}=s_i u_i$ and $\omega_{V_i}=s_i v_i$.
The intended normalization is $u_i+v_i=1$, so the decoded variable is
$x_i=u_i$ and the product $u_iv_i$ represents $x_i(1-x_i)$.
When the normalized prefix products equal one and the lower-scale
terms are omitted, the $U_i$ row is designed to impose
\begin{align*}
 0&=v_i-a_i-\sum_{j<i}b_{ij}u_jv_j
                  &&\text{(type L)},\\*
 0&=u_i(v_i-a_i)+c_i+\sum_{j>i}d_{ij}u_jv_j
                  &&\text{(type Q)}.
\end{align*}
Substituting $v_i=1-u_i$ gives exactly the source equations in
\cref{b:eq:linear,b:eq:quadratic}. In type Q, the displayed expression is
the payoff condition multiplied by $u_i$. Deleting $U_i$ from a later
prefix divides its product by $\omega_{U_i}=s_i u_i$, so this multiplication
cancels the resulting reciprocal factor. For type L, deleting $U_i$
leaves each earlier prefix intact. This explains the ordering of the
two equation types and the input terms in \cref{b:tab:payoffs}.
The rows for $C_i,V_i,Z_i$ enforce the normalization and product
constraints needed by this calculation. We next specify the payoff
tables and state the lemmas that justify the intended encoding.
\Cref{b:sec:correctness} proves these lemmas from the equilibrium
inequalities before deriving a solution to the quadratic system.

\paragraph{Payoff Functions.}

We now specify the payoff tables that implement these constraints.
The smallest reference prefix product determines their scale and
the required accuracy:
\begin{equation}\label{b:eq:precision-budget}
 S=S_{D+1},\qquad L=\frac{8}{S},\qquad \zeta=S^2,\qquad
 \eps_0=\frac{S^6}{2^{19}N}.
\end{equation}
Here $L$ bounds the unscaled payoff entries. The baseline $-\zeta$
will make each mixing player's coefficient contribution positive,
and $\eps_0$ makes the converted unscaled tolerance exactly $\zeta/8$.
\Cref{b:sec:correctness} verifies this accuracy and all the smallness thresholds.
Write $f_\ell(k)$ for the unscaled payoff difference.
For the auxiliary players, set $f_A(k)=c_0-k$ and $f_B(k)=k-T$.
For each encoding player, begin with $f_\ell(k)=-\zeta$ for every
$k$, then add the entries in \cref{b:tab:payoffs}. Entries with the same
player and count are added. The last three rows are included only for
the indicated equation type.

\begin{table}[!ht]
\centering
\renewcommand{\arraystretch}{1.1}
\begin{tabular}{@{}lll@{}}
\toprule
Player & Opponent Count $k$ & Term Added to $f_\ell(k)$\\
\midrule
$1$ & $2$ & $1/(\beta\chi_1)$\\
$1$ & $1$ & $-1/\beta$\\
$C_i$ & $r_i+1$ & $1/(\beta S_{r_i-1}s_i)$\\
$C_i$ & $r_i$ & $-1/(\beta S_{r_i-1})$\\
$V_i$ & $r_i+3$ & $1/(\beta S_{r_i}s_i\tau_i)$\\
$V_i$ & $r_i+2$ & $-1/(\beta S_{r_i}s_i)$\\
$Z_i$ & $r_i+4$ & $1/(\beta S_{r_i}K_i)$\\
$Z_i$ & $r_i+1$ & $-1/(\beta S_{r_i})$\\
$U_i$ & $r_i+2$ & $1/(\beta S_{r_i}s_i)$\\
$U_i$ & $r_i+1$ & $-a_i/(\beta S_{r_i})$\\
$U_i$, type L input $j<i$ & $r_j+3$ &
$-b_{ij}/(\beta S_{r_j+2})$\\
$U_i$, type Q input $j>i$ & $r_j+2$ &
$d_{ij}s_i/(\beta S_{r_j+2})$\\
$U_i$, type Q constant & $D+1$ &
$c_i s_i/(\beta S_{D+1})$\\
\bottomrule
\end{tabular}
\caption{Nonconstant Payoff Terms for the Encoding Players.
Every encoding table additionally has the constant term $-\zeta$.}
\label{b:tab:payoffs}
\end{table}

Only the $U_i$ rows depend on the source coefficients. The remaining
rows enforce the same support and normalization conditions for every
source instance.
Player $1$ and each $C_i,V_i,Z_i$ have two listed counts.
Each $U_i$ has two local counts, at most two input counts, and,
for type Q, one constant count. Thus every encoding player has
at most five exceptional counts; coinciding counts only reduce this
number.

The normalized payoffs are $\pi_\ell(0,k)=1/2$ and
$\pi_\ell(1,k)=1/2+f_\ell(k)/(2L)$.
Every listed count is at most $N-1$. The reference products have magnitude
$N^{-O(N^2)}$, and every unscaled payoff has magnitude at most $L=8/S$.
Thus the normalized payoffs lie in $[0,1]$, and their rational descriptions
have polynomial length. We give the bounds in \Cref{b:sec:induction}.

There is one useful feature of the signs and indices:
\begin{equation}\label{b:eq:prefix}
 f_\ell(k)\le-\zeta
 \quad\text{for every encoding player }\ell\text{ and every }k\le\ell.
\end{equation}
Indeed, every positive term occurs at a count strictly above the
player's index. This property will restrict the equilibrium support
before we use any approximate coefficient identities.

\paragraph{From a Payoff Row to a Coefficient Constraint.}
For example, the two entries for $C_i$ compare consecutive coefficients
at scale $s_i$. At a profile where all probabilities are less than one,
\cref{eq:two-strategy-generating} gives its expected unscaled payoff difference as
\[
 -\zeta+\frac{P_{0,C_i}}{\beta S_{r_i-1}}
 \left(\frac{e_{r_i+1}^{-C_i}}{s_i}-e_{r_i}^{-C_i}\right).
\]
The support argument in \cref{b:lem:auxiliary} establishes this probability
condition and shows that $C_i$ uses strategy $0$.
Its payoff condition then gives an upper bound on the parenthesized
expression, whether or not it mixes. If it mixes, the corresponding
lower bound is also available. The tolerance absorption in
\cref{b:eq:normal} turns these into the relative coefficient bounds
used below. The other control rows compare coefficients in the same way.

\subsection{Correctness of the Reduction}\label{b:sec:correctness}

The correctness proof uses two-sided payoff inequalities, so we must prove
that the encoding players mix. We first obtain coefficient bounds on
arbitrary supports, then use these bounds to force mixing in each variable
pair and pass the prefix normalization to the next block. Only after
this induction do we decode the source values.

\paragraph{Initial Support.}

We begin with an $\varepsilon_0$-$\ANE$, where
$\varepsilon_0$ is defined in \Cref{b:eq:precision-budget}.
By \Cref{lem:ane-wsne}, this converts to a well-supported equilibrium with
unscaled tolerance $\eta=\zeta/8$. The conversion threshold is verified in
\Cref{b:sec:induction}. In particular, $\eta<\zeta/4$ and
$\eta\le t^2/8<T$. All probabilities and weights below refer to this
converted profile.

Let $\mu=\sum_{\ell=1}^D p_\ell$ and write $a=p_A$, $b=p_B$.
The auxiliary players' expected unscaled differences are
$c_0-\mu-b$ and $\mu+a-T$.

The auxiliary players first bound the total probability of action $1$. This makes every encoding weight finite and gives an initial segment of active encoding players, before we assume that any encoding player mixes.

\begin{lemma}\label{b:lem:auxiliary}
Both $A$ and $B$ mix. Every encoding probability is
less than $T+\eta<1/4$, and the encoding players choosing $1$ with positive
probability form a nonempty initial segment of the player order.
Moreover, $|\mu+a-T|\le\eta$, $|\mu+b-c_0|\le\eta$, and
$P_{0,\ell}>1/5$.
\end{lemma}
\begin{proof}
If $\mu>T+\eta$, player $B$ must choose $1$ with probability one,
and then $A$ must choose $0$. Descend through the encoding players'
indices. Once all later encoding players choose $0$, player $\ell$
faces at most $\ell$ opponents choosing $1$.
By \cref{b:eq:prefix}, she must also choose $0$.
This gives $\mu=0$, a contradiction. Hence $\mu\le T+\eta$.

If $a=1$, then $B$ must choose $1$, which would force $A$ to choose
$0$. Thus $a<1$. Since $A$ uses strategy $0$, we obtain
$b\ge c_0-\mu-\eta\ge b_0-2\eta>0$.
If $b=1$, then $A$ chooses $0$, and the same descending argument
gives $\mu=0$, contradicting $B$'s best response.
It follows that $0<b<1$ and $|\mu+a-T|\le\eta$.

If $a=0$, the descending argument again forces every encoding
player to choose $0$, giving $\mu=0$ and contradicting
$|\mu+a-T|\le\eta$. Thus $A$ mixes.
Consequently $\mu\le T-a+\eta<T+\eta<1/4$ and
$|\mu+b-c_0|\le\eta$.

If precisely $m$ encoding players choose $1$ with positive
probability, any such player of index $\ell\ge m+1$ faces at most
$(m-1)+2\le\ell$ opponents choosing $1$. This contradicts
\cref{b:eq:prefix}. The $m$ indices must therefore be $1,\ldots,m$.
If $m=0$, the auxiliary equations give $a=T+O(\eta)$ and
$b=c_0+O(\eta)$. Player $1$ then has expected difference
$-\zeta+ab/(\beta t)-(a+b-2ab)/\beta>\eta$, another contradiction.

Finally, the inequality $\prod_j(1-z_j)\ge1-\sum_jz_j$ gives
\[
 P_{0,\ell}\ge\prod_j(1-p_j)
 \ge(1-\mu-a)(1-b)
 \ge(1-T-\eta)(1-c_0-\eta)>\frac{1}{5}. \qedhere
\]
\end{proof}

All weights are now finite. For an encoding player $\ell$, let $\D_\ell$
be the linear combination of $e_k^{-\ell}$ specified by the added
terms in \cref{b:tab:payoffs}. Her expected unscaled difference is
$-\zeta+P_{0,\ell}\D_\ell$. Every encoding player uses strategy $0$;
when $p_\ell>0$, she uses both strategies.
By \cref{b:lem:auxiliary}, with $\nu=5(\zeta+\eta)=O(\zeta)$,
\begin{equation}\label{b:eq:D}
 \D_\ell\le\nu\quad\text{for every encoding player},\qquad
 0<\D_\ell\le\nu\quad\text{if }p_\ell>0.
\end{equation}
The strict positivity follows from $\eta<\zeta$.

We next absorb this additive tolerance into small relative errors in
the coefficient inequalities. Set $\xi=t^{10}$ and write
$s_i'=s_i(1+\xi)$, $\tau_i'=\tau_i(1+\xi)$, and
$K_i'=K_i(1+\xi)$. The rows of $C_i,V_i,Z_i$ and player $1$ give
\begin{equation}\label{b:eq:normal}
 \begin{aligned}
 e_{r_i+1}^{-C_i}&\le s_i'e_{r_i}^{-C_i},&
 e_{r_i+3}^{-V_i}&\le\tau_i'e_{r_i+2}^{-V_i},\\
 e_{r_i+4}^{-Z_i}&\le K_i'e_{r_i+1}^{-Z_i},&
 e_2^{-1}&\le t(1+\xi)e_1^{-1}.
 \end{aligned}
\end{equation}
If the corresponding player has positive probability, the reverse
inequality holds with the unprimed target.
In that case, strict positivity in \cref{b:eq:D} makes the higher
coefficient positive, and hence the lower coefficient is positive
as well. Ratios used with these lower constraints are therefore
well-defined.

The additive tolerance can be absorbed without a lower bound on any individual positive probability. The proof also covers vanished coefficients; see \Cref{b:sec:cap}.

The upper bounds in \cref{b:eq:normal} hold even for players who choose
$0$ with probability one. This lets us bound the symmetric coefficients
before determining which variable pairs mix.

\paragraph{Coefficient Bounds before Mixing.}
The initial support argument gives upper inequalities for every encoding
player. We use them to control the symmetric coefficients without assuming
that the variable pairs mix.

Write $E_k=e_k(\omega)$ for the full weight list and define
\begin{equation}\label{b:eq:global-coefficients}
 M_i=\frac{E_{r_i+3}}{\beta S_{r_i}s_i^2},\qquad
 R=\frac{E_N}{\beta S_{D+1}}.
\end{equation}
For each block, use the notation
\[
 \begin{gathered}
 u_i=\frac{\omega_{U_i}}{s_i},\qquad v_i=\frac{\omega_{V_i}}{s_i},\qquad
 \Pi_i=\frac{\omega_B\prod_{\ell=1}^{r_i}\omega_\ell}{\beta S_{r_i}},\\
 H_i=\sum_{\ell=r_i+1}^{D+1}\omega_\ell,\qquad
 J_i=(\omega_\ell)_{\ell=r_i+4}^{D+1},\qquad
 T_i=\omega_{Z_i}+e_1(J_i).
 \end{gathered}
\]
Thus $H_i=s_i(u_i+v_i)+T_i$ is the total weight after $C_i$,
including $A$, and $\Pi_i$ is the normalized product of the prefix
through $C_i$, together with $B$.

The upper control inequalities hold on every support.
Deletion interlacing in \cref{b:lem:coefficients} and the
third-order inequality in \cref{b:lem:quarter-bound} give the
following estimate.

\begin{lemma}[A Coefficient Bound on Arbitrary Supports]\label{b:lem:cap}
For the converted WSNE specified above, if $Nt\le2^{-10}$, then
$M_i\le\frac14+16Nt$ for each $i\le q$ and $R\le1+64Nt$.
These bounds hold without assuming that the variable pairs mix.
\end{lemma}

Here $M_i$ sums products over all subsets of the prescribed size.
Its bound can force mixing before the coefficient is identified with
the intended pair's product. The control rows enforce the hypotheses
of the elementary coefficient bounds on every support,
independently of the source coefficients.

\begin{proof}[Proof Sketch]
Suppress the block index. Deletion interlacing transfers the $C,V$
upper constraints to $E_{r+2}\le s'E_{r+1}$ and
$E_{r+4}\le\tau'E_{r+3}$. Using only these bounds would lose a
factor $1/t$ when passing to the next prefix. We instead show
\[
 E_{r+5}\le(1+O(t))\tau K E_{r+1}.
\]
If $\omega_Z=0$, the $Z$ upper constraint gives
$E_{r+4}\le K'E_{r+1}$, and decreasing coefficient ratios give
$E_{r+5}\le\tau'E_{r+4}$.
If $\omega_Z>0$, the $Z$ lower constraint is available. After deleting
$V,Z$, let $g=e_{r+2}/e_{r+1}$ and $h=e_{r+3}/e_{r+2}$; strict
positivity of the $Z$ row makes these ratios well-defined.
Interlacing and the $Z$ lower constraint give $g\ge K/(s'\tau')$,
while the $V$ upper constraint gives
$h+\omega_Z\le\tau'(1+\omega_Z/g)$.
Since $\tau'/g\le s'(\tau')^2/K=O(t)$, we obtain
$h+\omega_Z\le\tau'(1+O(t))$, which gives the same prefix bound.
Neither case uses a lower constraint for $C$ or $V$.

Normalize by setting $A_i=E_{r_i+1}/(\beta S_{r_i})$ and
$A_{q+1}=R$. The preceding bound gives
$A_{i+1}\le(1+O(t))A_i$; the initial support constraints give
$A_1\le1+O(t)$. Hence $A_i,R\le1+O(Nt)$.
The same argument bounds the coefficient immediately before each prefix.
If $A_i\le1/8$, decreasing ratios already give
$M_i\le A_i(1+t^{10})^2<1/4$.
Otherwise, the bounded adjacent coefficients ensure that the largest
$r_i+1$ weights can be factored out with relative error $O(Nt)$.
For the remaining list $J$, the $C,V$ upper constraints become
$H=e_1(J)\le s_i(1+O(Nt))$ and
$e_3(J)\le\tau_i(1+O(Nt))e_2(J)$.
Since $\tau_i/s_i=t^2$, \cref{b:lem:quarter-bound} gives
$M_i\le1/4+O(Nt)$.
\Cref{b:sec:cap} supplies the explicit constants and the factorization
bounds, including the cases of zero coefficients.
\end{proof}

\paragraph{Compensation within One Block.}
The compensation illustrated in \Cref{b:fig:block} must control
each successive prefix product.
For $i<q$, the normalized prefix increment is exactly
\[
 \frac{\Pi_{i+1}}{\Pi_i}
 =\frac{\omega_{Z_i}}{\tau_i}
  \frac{u_iv_i\omega_{C_{i+1}}}{\chi_{i+1}}.
\]
The block step establishes $\omega_{Z_i}\approx\tau_i$ and
$u_iv_i\omega_{C_{i+1}}\approx\chi_{i+1}$, canceling this factor.
The next lemma states the hypotheses needed to pass these bounds
to the next block. Its constants depend only on the source margin:
\begin{equation}\label{b:eq:block-constants}
 \kappa=\frac{\delta}{4},\qquad \alpha=\frac{\delta}{16},\qquad
 \Lambda=32\kappa^{-1},\qquad \rho=\frac{\alpha}{8\Lambda}.
\end{equation}

\begin{lemma}[Mixing and Compensation within One Block]\label{b:lem:block-step}
Assume $Nt\le\rho$.
Fix $i\le q$ and suppose that
\begin{equation}\label{b:eq:block-entry}
 \begin{gathered}
 C_i\text{ mixes},\qquad
 \frac{\chi_i}{16}\le \omega_{C_i}\le4\kappa^{-1}\chi_i,\\
 \left|\frac{H_i}{s_i}-1\right|\le\alpha,\qquad |\Pi_i-1|\le\alpha,
 \end{gathered}
\end{equation}
and that the earlier blocks satisfy the coarse bounds
\begin{equation}\label{b:eq:earlier-block-bounds}
 \omega_{C_j}\ge\frac{\chi_j}{16},\qquad u_j,v_j\ge\kappa,\qquad
 \frac{\tau_j}{2}\le \omega_{Z_j}\le2\tau_j\quad(j<i).
\end{equation}
Then the following conclusions hold, with the same constant $\Lambda$
for every $i$.
\begin{enumerate}
\renewcommand{\labelenumi}{(\roman{enumi})}
\item The players $U_i,V_i,Z_i$ mix, $u_i,v_i\ge\kappa$,
$u_iv_i\ge\kappa/2$, and
\[
 \left|\frac{T_i}{\tau_i}-1\right|\le\Lambda t^2,\qquad
 \left|\frac{\omega_{U_i}\omega_{V_i}e_1(J_i)}{K_i}-1\right|\le\Lambda t.
\]
Moreover, $|\omega_{Z_i}/\tau_i-1|\le\Lambda t$.
\item If $i<q$, then $C_{i+1}$ mixes,
$\chi_{i+1}/16\le \omega_{C_{i+1}}\le4\kappa^{-1}\chi_{i+1}$, and
\begin{equation}\label{b:eq:block-output}
 \left|\frac{H_{i+1}}{s_{i+1}}-1\right|\le\Lambda Nt,\qquad
 \left|\frac{\Pi_{i+1}}{\Pi_i}-1\right|\le\Lambda t.
\end{equation}
\item If $i=q$, then $|R/\Pi_q-1|\le\Lambda t$.
\end{enumerate}
\end{lemma}

\begin{proof}[Proof Sketch]
We establish support in the order $U_i,V_i,Z_i,C_{i+1}$.
The entry hypotheses and scale separation allow us to factor out the
prefix through $C_i$. The two local terms in the $U_i$ row become
$\Pi_i(H_i/s_i-u_i-a_i)+O(Nt^3)$.
For type L, each subtracted input coefficient is at most
$1/4+O(Nt)$ by \cref{b:lem:cap}, and the source margin gives
$1-a_i-\tfrac14\sum_j b_{ij}\ge\delta/2$.
For type Q, every input term is nonnegative and $a_i=9/10$.
In either case the upper payoff inequality forces $u_i\ge\kappa$.
Thus $U_i$ mixes, and its lower inequality is now available.

Before using a lower inequality for $V_i$, we use its upper
constraint to bound the weight after the pair:
\[
 \omega_{U_i}T_i
 \le\tau_i(1+O(Nt^3))(\omega_{U_i}+T_i).
\]
Since $\omega_{U_i}\ge\kappa s_i$ and $\tau_i/s_i=t^2$,
this yields $T_i\le\tau_i(1+O(t^2))$.
The local $U_i$ expression is consequently
$\Pi_i(v_i-a_i)+O(t^2)$.
For type L, its lower inequality gives $v_i\ge a_i-O(t^2)>\kappa$.
For type Q, multiply that inequality by $u_i$ and bound each input
by its full-list coefficient to obtain
\[
 \Pi_i u_i(a_i-v_i)
 \le c_iR+\sum_{j>i}d_{ij}M_j+O(t^2)
 \le a_i-\delta+O(Nt).
\]
If $v_i<\kappa$, then $u_i\ge1-\alpha-\kappa-O(t^2)$, so the
left-hand side is at least $a_i-2\alpha-2\kappa-O(t^2)$.
This contradicts $2\alpha+2\kappa=5\delta/8$ once the error terms
are absorbed. Hence $V_i$ mixes; its lower constraint now gives
$T_i/\tau_i=1+O(t^2)$.

The bounds on the pair's sum and its two entries give
$u_iv_i\ge\kappa/2$.
If $\omega_{Z_i}=0$, then $e_1(J_i)=T_i\ge\tau_i/2$.
The $Z_i$ upper constraint would require
$e_3(\omega_{U_i},\omega_{V_i},J_i)\le2K_i$, but this coefficient
is at least $\kappa s_i^2\tau_i/4>2K_i$, since
$K_i=s_i^2\tau_i t$. Thus $Z_i$ mixes.
Its two-sided constraint gives
$e_3(\omega_{U_i},\omega_{V_i},J_i)=K_i(1+O(Nt^3))$.
Expanding this coefficient yields
\[
 e_3(\omega_{U_i},\omega_{V_i},J_i)
 =\omega_{U_i}\omega_{V_i}e_1(J_i)+O(s_i\tau_i^2),
 \qquad s_i\tau_i^2=tK_i.
\]
Therefore $\omega_{U_i}\omega_{V_i}e_1(J_i)=K_i(1+O(t))$,
$e_1(J_i)=\chi_{i+1}(1+O(t))/(u_iv_i)$, and
$\omega_{Z_i}/\tau_i=1+O(t)$.

For $i<q$, if $C_{i+1}$ had weight zero, its coefficient ratio would
be $\Omega(\chi_{i+1})$, contradicting the upper target
$s_{i+1}=t^4\chi_{i+1}$. Hence it mixes. Its two-sided row gives
$H_{i+1}/s_{i+1}=1+O(Nt)$ directly, using the preceding estimates
only as coarse bounds. It also gives
$\omega_{C_{i+1}}/e_1(J_i)=1+O(t^3)$, and the exact identity
\[
 \frac{\Pi_{i+1}}{\Pi_i}
 =\frac{\omega_{Z_i}}{\tau_i}
  \frac{\omega_{U_i}\omega_{V_i}e_1(J_i)}{K_i}
  \frac{\omega_{C_{i+1}}}{e_1(J_i)}
 =1+O(t)
\]
proves the compensation. The total-weight estimate is obtained anew
at each block; only the prefix-product errors accumulate.
For $i=q$, use $J_q=(\omega_A)$ and the mixing of $A$ from
\cref{b:lem:auxiliary}. The same product identity gives the terminal
bound, with no next control player.
\Cref{b:sec:induction} verifies the constants in all these estimates.
\end{proof}
\paragraph{Propagating Mixing across the Blocks.}

We now propagate the block lemma.
Set $A_{\mathrm{norm}}=2\Lambda=2^{15}$.
Our common threshold satisfies $\rho=\alpha/(4A_{\mathrm{norm}})$.
The next lemma makes both the accumulated prefix error and the
total-weight error explicit.

\begin{lemma}[Mixing and Normalization]\label{b:lem:normalization}
If $Nt\le\rho$, then every player mixes in the
WSNE under consideration. For each $i\le q$,
\begin{align*}
 \frac{\chi_i}{16}\le \omega_{C_i}&\le4\kappa^{-1}\chi_i,
 &u_i,v_i&\ge\kappa,\\*
 \left|\frac{H_i}{s_i}-1\right|&\le A_{\mathrm{norm}}Nt,
 &|u_i+v_i-1|&\le A_{\mathrm{norm}}Nt,\\*
 \left|\frac{\omega_{Z_i}}{\tau_i}-1\right|&\le\Lambda t,
 &|\Pi_i-1|&\le A_{\mathrm{norm}}it.
\end{align*}
In addition, $|R-1|\le A_{\mathrm{norm}}Nt$.
\end{lemma}
\begin{proof}
\step{Initializing the First Block.}
By \cref{b:eq:initial-weights}, $\omega_B\ge1$, $\omega_1\ge1/14$, and the
remaining total weight $H$ lies in $[t,t(1+6t)]$.
Factoring out $(\omega_B,\omega_1)$ in \cref{b:eq:dominance} has relative
error at most $60t$.
If $p_{C_1}=0$, the ratio $e_3^{-C_1}/e_2^{-C_1}$ is therefore at
least $t/(1+60t)\ge t/2>s_1'$, a contradiction.
Thus $C_1$ mixes, and its two-sided constraint gives
$\frac{1}{1+60t}\le\frac{H_1}{s_1}\le(1+\xi)(1+60t)$.
It follows that
\[
 \left|\frac{H_1}{s_1}-1\right|\le64t,\qquad
 \left|\frac{\omega_{C_1}}{\chi_1}-1\right|\le8t,\qquad |\Pi_1-1|\le256t.
\]
For the second bound use $\omega_{C_1}=H-H_1$ and $s_1/\chi_1=t^4$.
For the third, the three factors
$\omega_B/\beta$, $7\omega_1$, and $\omega_{C_1}/\chi_1$ have errors at most
$48t$, $28t$, and $8t$, respectively.
These bounds establish \cref{b:eq:block-entry} for $i=1$.

\step{Passing the Bounds to the Next Block.}
We prove \cref{b:eq:block-entry} by induction, with the initialization as
the base case. For $i>1$, suppose it holds for blocks $1,\ldots,i-1$.
Their block steps supply \Cref{b:eq:earlier-block-bounds},
mixing and the weight bounds for $C_i$, and $|H_i/s_i-1|\le\Lambda Nt$.
The prefix $\Pi_i$ is the initial prefix times the successive block
increments. Their total relative error is at most
$[256+(i-1)\Lambda]t\le\Lambda it<1/2$.
The elementary product bound
$|\prod_j(1+z_j)-1|\le2\sum_j|z_j|$ when $\sum_j|z_j|\le1/2$ gives
$|\Pi_i-1|\le A_{\mathrm{norm}}it$.
Both this bound and $\Lambda Nt$ are below $\alpha$, by
\Cref{b:eq:common-smallness}. Thus the hypotheses hold at block $i$ as well.
The induction proves that every player mixes.

The initialization and \cref{b:eq:block-output} bound each
$H_i/s_i-1$ directly. Since
$u_i+v_i=H_i/s_i-T_i/s_i$ and $T_i/s_i\le2t^2$,
the stated bound on $u_i+v_i$ follows.
The bound on $\Pi_i$ is the one just proved.
Finally, include the terminal factor $R/\Pi_q$ in the same product bound.
By \cref{b:lem:block-step}(iii), this gives
$|R-1|\le2[256+q\Lambda]t \le A_{\mathrm{norm}}Nt$.
This proves all the asserted estimates.
\end{proof}

\paragraph{From Equilibria to Solutions of the Quadratic System.}

By \cref{b:lem:normalization}, every encoding player now satisfies the
two-sided bound in \cref{b:eq:D}. Set
$A_{\mathrm{dec}}=16A_{\mathrm{norm}}=2^{19}$.
For our fixed $\kappa$, the constants in the local estimates satisfy
$32\kappa^{-1}+802\le A_{\mathrm{norm}}$.
We show that every decoded source equation has residual at most
$A_{\mathrm{dec}}Nt$. This constant is fixed independently of
$N$, the block index, and the source instance.

Define the decoded variable
$x_i=\omega_{U_i}/(\omega_{U_i}+\omega_{V_i})=u_i/(u_i+v_i)$.
By \cref{b:lem:normalization}, it is well-defined and belongs to $[0,1]$.
Writing $h_i=u_i+v_i$, we have
$u_i=x_ih_i$, $v_i=(1-x_i)h_i$, and $h_i\le2$.
Thus
$|u_i-x_i|,\quad |v_i-(1-x_i)|,\quad |u_iv_i-x_i(1-x_i)|\le A_{\mathrm{norm}}Nt$.

We now evaluate the actual input terms. The prefix through $V_j$,
together with $B$, has minimum weight at least $\kappa s_j$, and the
remaining total weight is at most
$2\tau_j$. Applying \cref{b:eq:dominance} yields
\begin{align}
 \frac{e_{r_j+3}^{-U_i}}{\beta S_{r_j+2}}
 &=\Pi_j u_jv_j(1+\theta_{ij})
                         &&(j<i),\label{b:eq:read-L}\\*
 \frac{s_i e_{r_j+2}^{-U_i}}{\beta S_{r_j+2}}
 &=\frac{\Pi_j u_jv_j}{u_i}(1+\theta_{ij})
                         &&(j>i),\label{b:eq:read-Q}\\*
 \frac{s_i e_{N-1}^{-U_i}}{\beta S_{D+1}}
 &=\frac{R}{u_i},\label{b:eq:read-constant}
\end{align}
where $0\le\theta_{ij}\le4\kappa^{-1}Nt^2$.
Indeed, the sum of the tail-to-prefix ratios is at most
$2\kappa^{-1}Nt^2$; \cref{b:eq:dominance} doubles this bound.
The designated prefix through $V_j$, together with $B$, contains
$r_j+3$ entries. For $j<i$, deleting $U_i$ removes an entry after
this prefix, so the coefficient of degree $r_j+3$ in
\cref{b:eq:read-L} selects the entire prefix.
For $j>i$, the deleted player belongs to the prefix, leaving
$r_j+2$ entries. Its product is divided by $\omega_{U_i}=s_i u_i$;
the factor $s_i$ in \cref{b:eq:read-Q} leaves the reciprocal $1/u_i$.
\Cref{b:eq:read-constant} is an exact identity for the full product.

For a type L equation, combine the local expression in
\cref{b:eq:local-row} with the input terms in \cref{b:eq:read-L}.
Using $\nu\le Nt^3$, $|\lambda_i|\le400t^2$,
$|\Pi_j-1|\le A_{\mathrm{norm}}Nt$, and $\sum_jb_{ij}<2$ gives
\[
 \left|v_i-a_i-\sum_{j<i}b_{ij}u_jv_j\right|
 \le6A_{\mathrm{norm}}Nt.
\]
Substituting the decoded variables adds at most
$3A_{\mathrm{norm}}Nt$ to this bound, so
\begin{equation}\label{b:eq:residual-L}
 \left|x_i-1+a_i+\sum_{j<i}b_{ij}x_j(1-x_j)\right|
 \le A_{\mathrm{dec}}Nt.
\end{equation}

For a type Q equation, multiply the two-sided bound on $\D_{U_i}$
by $u_i$, then use \cref{b:eq:read-Q,b:eq:read-constant}.
Using $u_i\le2$, $|R-1|\le A_{\mathrm{norm}}Nt$, and
$c_i+\sum_jd_{ij}<4$ gives
\[
 \left|u_i(v_i-a_i)+c_i+\sum_{j>i}d_{ij}u_jv_j\right|
 \le8A_{\mathrm{norm}}Nt.
\]
Multiplying before substitution cancels the reciprocal $1/u_i$ and keeps the error bounded by the source coefficient margins.
Substitution adds at most $6A_{\mathrm{norm}}Nt$.
Since $x_i(a_i-(1-x_i))=x_i^2-(1-a_i)x_i$, we obtain
\begin{equation}\label{b:eq:residual-Q}
 \left|x_i^2-(1-a_i)x_i-c_i-\sum_{j>i}d_{ij}x_j(1-x_j)\right|
 \le A_{\mathrm{dec}}Nt.
\end{equation}

For $t=2^{-36}/N$, these residuals are at most $\varepsilon_{\mathrm{src}}$,
and $N^{-54N^2}$ is sufficiently accurate for the initial conversion.
\Cref{b:sec:induction} verifies these inequalities and the polynomial
running time of the solution map. This proves \Cref{b:thm:core}.

\section{FIXP-Hardness with Three Strategies}\label{sec:three-strategy-overview}

We prove \Cref{t:thm:quantitative-reduction} using the circuit simulation
outlined in \Cref{sec:three-strategy-principle}.
We first prescribe the signal probabilities, then construct the
gate payoffs from root statistics and establish the supports needed
for exact decoding. Throughout this section, $w_P$ and $\theta_P$
denote the signal probability and conditional value defined there.

\subsection{Enforcing the Signal Probabilities}\label{t:sec:calibration}

We next turn the near-Fekete nodes into equilibrium constraints.  The common calibration payoff depends only on total signal probabilities.  Later gate terms perturb only one action and are scaled so that they cannot destroy the global calibration bounds.

Choose an even population $N$ and distinct rational nodes $0<\delta_1<\cdots<\delta_N<1$. We use rational approximations to the Fekete points of $[0,1]$, which maximize $V(y)=\prod_{i<j}(y_j-y_i)$. The nodes need only polynomial precision: their distance from the endpoints is $N^{-O(1)}$, every gap is $\Omega(N^{-2})$, and a central set of $\Omega(N)$ consecutive nodes has spacing $\Theta(1/N)$. These bounds follow from the extremal property of Fekete points; \Cref{t:sec:fekete} gives the short argument and construction. Put $b=\min_i\{\delta_i,1-\delta_i\}$ and $g=\min_{i<N}(\delta_{i+1}-\delta_i)$.

Let $K=h_s+h_c$ count the opponents choosing $s$ or $c$. For each player $i$, define
$c_i:=\prod_{j\ne i}(\delta_j-\delta_i)$, $c_*:=\min_i|c_i|$,
and
\begin{equation*}
 \Phi_i(K):=(1-\delta_i)^K(-\delta_i)^{N-1-K},
 \qquad
 G_i(K):=|c_i|-\operatorname{sgn}(c_i)\Phi_i(K).
\end{equation*}
Both $s$ and $c$ have base payoff $G_i(K)$, while $t$ has payoff zero. We perturb only $c$ to express the gate equations. Averaging the calibration table produces a product whose factors distinguish the players.

\begin{lemma}[Expected Calibration Advantage]\label{t:lem:expected-calibration}
For every mixed profile,
$\E[G_i(K)]=|c_i|(1-\varrho_i(w))$.
Moreover every entry of $G_i$ lies in $[-2,2]$.
\end{lemma}

\begin{proof}
Each opponent $j$ contributes an independent Bernoulli signal indicator of mean $w_j$.  Therefore
\[
 \E[\Phi_i(K)]
 =\prod_{j\ne i}\bigl(w_j(1-\delta_i)+(1-w_j)(-\delta_i)\bigr)
 =Q_i(w).
\]
The formula follows because $\operatorname{sgn}(c_i)Q_i=|c_i|\varrho_i$.  Finally $|c_i|\le1$ and $|\Phi_i(K)|\le1$.
\end{proof}

\paragraph{Conservation and Global Localization.}

The ratio map satisfies a conservation law that couples its coordinates even when some signal mass lies on the boundary.  This identity turns the one-sided equilibrium inequalities into two-sided ratio control.

\begin{lemma}[Lagrange Conservation]\label{t:lem:conservation}
For every $w\in\R^N$,
\begin{equation}\label{t:eq:conservation}
 \sum_{i=1}^N(w_i-\delta_i)(\varrho_i(w)-1)=0.
\end{equation}
\end{lemma}

\begin{proof}
Let
\[
 P_w(z)=\prod_{j=1}^N(w_j-z),
 \qquad
 P_\delta(z)=\prod_{j=1}^N(\delta_j-z),
 \qquad
 H(z)=P_w(z)-P_\delta(z).
\]
The leading terms cancel, so $\deg H\le N-1$.  At $z=\delta_i$,
$H(\delta_i)=(w_i-\delta_i)Q_i(w)$.
Interpolate $H$ at the nodes $\delta_i$ and compare the coefficient of $z^{N-1}$.  Since
\[
 \prod_{j\ne i}(\delta_i-\delta_j)=(-1)^{N-1}c_i,
\]
we obtain
$\sum_i(w_i-\delta_i) =\sum_i(w_i-\delta_i)\frac{Q_i(w)}{c_i}$,
which is \cref{t:eq:conservation}.
\end{proof}

The gate payoffs will perturb only action $c$.  We therefore analyze calibration under an arbitrary histogram-dependent perturbation of $c$ whose expected magnitude is at most $\gamma$, while $s$ and $t$ remain unperturbed.

\begin{lemma}[Ratio Control on Arbitrary Supports]\label{t:lem:ratio-control}
Suppose that the payoff of action $c$ receives an arbitrary histogram-dependent perturbation whose expected magnitude is at most $\gamma$, while $s$ and $t$ remain unperturbed.  Let $\xi\ge0$.  At every $\xi$-$\WSNE$ of the resulting game,
\begin{equation*}
 |\varrho_i(w)-1|
 \le
 \frac{2N(\gamma+\xi)}{c_*b}
 \qquad\text{for every }i.
\end{equation*}
\end{lemma}

\begin{proof}
Write $T=(\gamma+\xi)/c_*$ and $a_i=(w_i-\delta_i)(\varrho_i-1)$.  Comparing a supported action with $s$ or $t$ gives
\[
 \begin{array}{ll}
 \varrho_i-1\ge-T,&w_i=0,\\
 \varrho_i-1\le T,&w_i=1,\\
 |\varrho_i-1|\le T,&0<w_i<1.
 \end{array}
\]
Indeed, at zero signal mass $t$ is supported; at full signal mass some signal action is supported; and at an interior mass both types are supported.  Removing the perturbation of $c$ changes a signal payoff by at most $\gamma$.

In all three cases $a_i\le T$.  Conservation gives $\sum_i a_i=0$, hence $a_i\ge-(N-1)T$ and $|a_i|\le NT$.  At a boundary mass, $|w_i-\delta_i|\ge b$, so $|\varrho_i-1|\le NT/b$.  At an interior mass the stronger direct bound $T$ already holds.  Both imply the stated estimate.
\end{proof}

Ratio control is still an implicit condition on the product map. We convert it into coordinate localization by comparing a product of ratios with the Lagrange interpolation polynomials, then using the near-optimality of the Vandermonde product.

For ordered nodes $a_1<\cdots<a_N$, let
\[
 \ell_i^a(z)=\prod_{j\ne i}\frac{z-a_j}{a_i-a_j}.
\]

\begin{lemma}[Product Factorization]\label{t:lem:product-factorization}
For every $w$,
\begin{equation*}
 \prod_{i=1}^N\varrho_i(w)
 =\prod_{i=1}^N\ell_i^\delta(w_i).
\end{equation*}
\end{lemma}

\begin{proof}
Both sides are the product over all ordered pairs $i\ne j$ of
$\frac{w_j-\delta_i}{\delta_j-\delta_i}$.
\end{proof}

The Fekete choice now gives a quantitative separation. For nodes sufficiently
close to the Fekete maximizer, \Cref{t:lem:near-fekete-gap} shows that if any
$|w_i-\delta_i|\ge\rho_0$, then
$|\prod_i\ell_i^\delta(w_i)|\le\exp(-\rho_0^2/(32N))$.
Here $\rho_0=N^{-O(1)}$ is chosen small enough for the local inverse
below. This gap follows from the strict concavity of $\log V$; its quantitative
proof is in \Cref{t:sec:fekete}.

\paragraph{Local Inversion.}

At the target, the calibration Jacobian is
\begin{equation}\label{t:eq:main-jacobian}
 J_{ii}=0,\qquad J_{ij}=\frac1{\delta_j-\delta_i}\quad(i\ne j).
\end{equation}
By \Cref{t:lem:elementary-inverse}, the inverse norm depends only
on the length of the node interval, even when some nodes are very close.

Write $R(w)=(\varrho_1(w),\ldots,\varrho_N(w))$. Its Jacobian at $\delta$ is the matrix in \Cref{t:eq:main-jacobian}, so \Cref{t:lem:elementary-inverse} gives
$\|J^{-1}\|_\infty\le\sqrt N\le N$.
The local estimate uses the normalized factors directly. Their distances
from zero are controlled by the node gap $g$, rather than the much smaller
product $c_*$. This keeps the localization radius at polynomial scale.

\begin{lemma}[Local Inverse Inequality]\label{t:lem:local-inverse}
If $\rho_0\le g^2/(4N^3)$ and
$\|w-\delta\|_\infty<\rho_0$, then
$\|w-\delta\|_\infty\le2N\|R(w)-\one\|_\infty$.
\end{lemma}

\begin{proof}
Put $v=w-\delta$ and $h=\|v\|_\infty$. From \Cref{t:eq:ratio-map},
the normalized product is
\[
 \varrho_i(w)=\prod_{j\ne i}
 \left(1+\frac{v_j}{\delta_j-\delta_i}\right).
\]
The constant and linear terms are $1$ and $(Jv)_i$.
The absolute sum of the remaining terms is at most
$\exp(Nh/g)-1-Nh/g\le N^2h^2/g^2$, since $Nh/g\le1$.
Thus $\|R(w)-\one-Jv\|_\infty\le N^2h^2/g^2\le h/(4N)$.
Since $\|J^{-1}\|_\infty\le N$, we obtain
$\|v\|_\infty\le N\|R(w)-\one\|_\infty+\|v\|_\infty/4$.
Rearranging proves the claim.
\end{proof}

We can now complete the proof of the global calibration theorem,
using the product gap and then the local inverse.

\begin{proof}[Proof of \Cref{t:thm:robust-localization}]
By \Cref{t:lem:ratio-control},
\[
 |\varrho_i-1|
 \le \varepsilon_\rho:=\frac{2N(\gamma+\xi)}{c_*b}
 \le\frac{\rho_0^2}{256N^2}<\frac12.
\]
Thus every $\varrho_i$ is positive.  Since $\log(1-z)\ge-2z$ for $0\le z\le1/2$,
$\log\prod_i\varrho_i \ge-2N\varepsilon_\rho \ge-\frac{\rho_0^2}{128N}$.
If some coordinate were $\rho_0$-far from its node,
\Cref{t:lem:product-factorization} and the assumed product gap would instead imply
$\log\prod_i\varrho_i \le-\frac{\rho_0^2}{32N}$,
a contradiction.

Now \Cref{t:lem:local-inverse} applies inside the $\rho_0$-box.
Combining it with \Cref{t:lem:ratio-control} gives
$\|w-\delta\|_\infty\le2N\|R(w)-\one\|_\infty
\le4N^2(\gamma+\xi)/(c_*b)$.
\end{proof}

The global and local estimates serve different purposes.  During the support analysis, the gate perturbation $\gamma$ is part of the coarse error bound.  After support is known, comparison of the two unperturbed actions eliminates $\gamma$ and makes the signal displacement proportional to $\xi$ alone.
\subsection{Simulating Circuit Gates}\label{t:sec:semantic-framework}

Calibration distinguishes players by their signal probabilities. We now
use the corresponding roots to recover conditional values and implement
the circuit gates. First-order statistics recover a single value;
second-order statistics, together with four auxiliary players, produce a
product. We encode source values in an interior interval to make the
subsequent support argument possible.

\paragraph{Expressing Individual Probabilities Using Polynomial Roots.}

The two recovery statistics in \Cref{fig:root-recovery} come from differentiating the histogram generating function. Fix an observer $o$ at prescribed signal probabilities $w=\delta$. In its histogram let $K=h_s+h_c$ and $L=h_c$. For opponent $j$, put
\[
 A_j(z)=1-\delta_j+\delta_jz,\qquad
 \zeta_j=1-1/\delta_j,\qquad
 N_{o,j}=\delta_j\prod_{\ell\ne o,j}A_\ell(\zeta_j).
\]
The normalizer $N_{o,j}$ is nonzero. We will use the first-order identity
\begin{equation*}
 \E_\delta[Lz^{K-1}]
 =\sum_{j\ne o}\delta_j\theta_j\prod_{\ell\ne o,j}A_\ell(z).
\end{equation*}
At $z=\zeta_j$, every summand but $j$ vanishes. Consequently,
\begin{equation*}
 T^{(1)}_{o, j}(K,L)=\frac{L\zeta_j^{K-1}}{N_{o,j}},
 \qquad \E_\delta[T^{(1)}_{o, j}]=\theta_j,
\end{equation*}
where the table is zero at $K=0$. Thus an anonymous payoff table recovers opponent $j$'s conditional probability exactly.

At second order,
\begin{equation*}
 \E_\delta\!\left[\binom L2z^{K-2}\right]
 =\sum_{\substack{p<q\\p,q\ne o}}\delta_p\theta_p\delta_q\theta_q
   \prod_{\ell\ne o,p,q}A_\ell(z).
\end{equation*}
Evaluating at $\zeta_B$ and dividing by $N_{o,B}$ gives $T^{(2)}_{o,B}(K,L)=\binom L2\zeta_B^{K-2}/N_{o,B}$ for $K\ge2$, and zero otherwise. Its expectation is
\begin{equation}\label{t:eq:main-star}
 \E_\delta[T^{(2)}_{o,B}]
 =\theta_B\sum_{p\ne o,B}\frac{\theta_p}{\zeta_B-\zeta_p}.
\end{equation}
Every recovery statistic used by a gate has polynomially described rational
entries and magnitude at most $2^{O(N)}$: its root player lies in the central
band. For any table $T$ bounded by $B$, coupling profiles with the same
conditional probabilities gives
\begin{equation}\label{t:eq:main-reader-stability}
 |\E_w[T(K,L)]-\E_\delta[T(K,L)]|\le2BN\|w-\delta\|_\infty.
\end{equation}
\Cref{t:lem:reader-bits} proves the bounds on their rational descriptions. The first index of $T^{(1)}$ or $T^{(2)}$ identifies the observer; the second identifies the player whose root is used. We suppress $(K,L)$ in payoff-table formulas.

\paragraph{Spacing for Multiplication.}

We place the players used by gate equations among the central nodes, leaving
$\Theta(\log(2G))$ filler players between successive groups. A filler has
$d_P=-1$, and the support argument will force its conditional probability to
zero. For each bilinear gate, the two players $V_i,B_i$ occupy adjacent nodes.
Put $h_i=\zeta_{B_i}-\zeta_{V_i}>0$ and $d_i=1/h_i=\Theta(N)$.
Writing $\calA$ for the non-filler players, this arrangement gives
\begin{equation}\label{t:eq:main-sparse}
 h_i\le1/64,\qquad
 \sum_{p\in\calA\setminus\{B_i,V_i\}}
 \left|\frac{h_i}{\zeta_{B_i}-\zeta_p}\right|\le1/64.
\end{equation}
Indeed, the contribution of a group at distance $r$ is $O(1/(rL))$ when
$L$ fillers separate groups; summing over $O(G)$ groups gives $O(\log(2G)/L)$.
Choosing a sufficiently large constant in $L$ gives the displayed bounds.
The population is $N=\Theta(G\log(2G))$. \Cref{t:sec:sparse-details}
gives a rational construction and the explicit estimates.

\paragraph{Protected Gate Values.}

The recovery statistics give each observer the conditional probabilities needed by its gate.  To keep the support argument away from the boundary of the unit interval, we encode every source value inside a protected interval with a constant margin on both sides.

For a source value $x\in[0,1]$, use the protected encoding
$E(x)=\frac14+\frac x2$.
Thus $u=E(x)\in[1/4,3/4]$ and $x=2u-1/2$.  The protected raw expressions are
\begin{equation*}
 f_i(u)=
 \begin{cases}
 1-u_j,&G_{1-},\\
 2u_j-1/4,&G_{\times2},\\
 1/4+\tfrac12(u_j+1/4)(u_k+1/4),&G_\phi.
 \end{cases}
\end{equation*}
At the prescribed signal probabilities we require $u_i=\mathcal P(f_i(u))$.

A gate expression may leave the protected interval.  We therefore apply
\begin{equation*}
 \mathcal P(z)=\min\left\{\frac34,\max\left\{\frac14,z\right\}\right\}
\end{equation*}
in two stages. The lower stage implements the maximum after an affine change of coordinates that fixes $1/4$, and the upper stage implements the final minimum. This change keeps the lower-stage target strictly inside $[0,1]$ even before the gate equations hold, which is essential for establishing support.

\paragraph{Gate Players.}

For every source gate $i$, create six active players
$Z_i,\ X_i,\ L_i^0,\ L_i^F,\ U_i^z,\ U_i^b$.
Write
\[
 r_i=\theta_{Z_i},
 \qquad
 u_i=\theta_{X_i},
 \qquad
 \lambda_i^0=\theta_{L_i^0},
 \qquad
 \lambda_i^F=\theta_{L_i^F},
\]
and
$\mu_i^z=\theta_{U_i^z}, \qquad \mu_i^b=\theta_{U_i^b}$.
Define
\begin{equation*}
 \calR(z)=\frac27z+\frac5{28},
 \qquad
 \calR^{-1}(r)=\frac72r-\frac58,
 \qquad
 z_i:=\calR^{-1}(r_i).
\end{equation*}

For each bilinear gate $i$ create four additional active players
$B_i,\ C_i,\ V_i,\ W_i$.
Write
\[
 \beta_i=\theta_{B_i},
 \quad
 \chi_i=\theta_{C_i},
 \quad
 \nu_i=\theta_{V_i},
 \quad
 \psi_i=\theta_{W_i}.
\]
The adjacent block $V_i,B_i$ has the spacing specified in \cref{t:eq:main-sparse}.  The unique external observer that needs the product is
$O_i:=L_i^F$.

\begin{table}[t]
\centering
\small
\begin{tabular}{@{}ll@{}}
\toprule
Player(s) & Equilibrium Condition \\
\midrule
$X_i$ & represents the output $u_i$ \\
$Z_i$ & represents $\calR(z_i)$ for the lower projection \\
$L_i^0,L_i^F$ & implement the lower maximum \\
$U_i^z,U_i^b$ & implement $\min\{z_i,3/4\}$ \\
$B_i,C_i$ & copy the first input of a bilinear gate \\
$V_i,W_i$ & cancel the terms other than the second input \\
\bottomrule
\end{tabular}
\caption{Gate Players Associated with Source Gate $i$.  The last two rows are present only for a bilinear gate.}
\label{t:tab:gate-players}
\end{table}

\paragraph{A Multiplication Construction with Four Players.}

Fix a bilinear gate $i$ with inputs $j,k$.  Define
\begin{equation}\label{t:eq:Si}
 S_i
 :=
 \sum_{p\in\calA\setminus\{O_i,B_i,V_i,W_i\}}
 \frac{\theta_p}{\zeta_{B_i}-\zeta_p},
\end{equation}
$\overline R_i := \frac{1/2}{\zeta_{B_i}-\zeta_{W_i}}+S_i$,
and
$F_i^{\mathrm{mul}} := \frac12+h_i(u_k-\overline R_i)$.
Because the support analysis forces every filler to have $\theta=0$, they are absent from \cref{t:eq:Si}; the observer $O_i$ is absent because the second-order statistic does not see its own player.

The design is easiest to see at the prescribed nodes.  The four auxiliary equations will enforce
\[
 \chi_i=\frac12,
 \qquad
 \beta_i=u_j,
 \qquad
 \psi_i=\frac12,
 \qquad
 \nu_i=F_i^{\mathrm{mul}}.
\]
Since $d_i h_i=1$, the last two equalities imply
\begin{equation}\label{t:eq:mult-cancellation-overview}
 d_i\left(\nu_i-\frac12\right)
 +\frac{\psi_i}{\zeta_{B_i}-\zeta_{W_i}}+S_i
 =u_k.
\end{equation}
The second-order recovery statistic multiplies the left-hand side by $\beta_i$, producing $u_j u_k$.  The residuals below implement these four equalities using only the conditional probabilities of distinct opponents.

The ideal gate residuals for the four auxiliary players are
\begin{align*}
D_{B_i}&=\chi_i-\frac12,\\
D_{C_i}&=u_j-\beta_i,\\
D_{V_i}&=\psi_i-\frac12,\\
D_{W_i}&=F_i^{\mathrm{mul}}-\nu_i.
\end{align*}
We obtain these residuals from the following anonymous payoff tables:
\begin{align*}
d_{B_i}
&=T^{(1)}_{B_i, C_i}-\frac12,\\
d_{C_i}
&=T^{(1)}_{C_i, X_j}
 -T^{(1)}_{C_i, B_i},\\
d_{V_i}
&=T^{(1)}_{V_i, W_i}-\frac12,\\
d_{W_i}
&=\frac12
 +h_iT^{(1)}_{W_i, X_k}
 -\frac{h_i}{2(\zeta_{B_i}-\zeta_{W_i})}
 -T^{(1)}_{W_i, V_i}\\*
&\hspace{8mm}
 -\sum_{p\in\calA\setminus\{O_i,B_i,V_i,W_i\}}
 \frac{h_i}{\zeta_{B_i}-\zeta_p}
 T^{(1)}_{W_i, p}.
\end{align*}
The table $d_{W_i}$ replaces the unavailable self-term involving $\psi_i$ by the constant $1/2$.  The auxiliary $V_i,W_i$ cycle forces $\psi_i=1/2$ at zero error and within the gate error otherwise.

For the external observer $O_i=L_i^F$, define the product table
\begin{equation*}
 \Pi_i
 :=
 T^{(2)}_{O_i,B_i}
 -\frac{d_i}{2}T^{(1)}_{O_i, B_i}.
\end{equation*}
At prescribed nodes, \cref{t:eq:main-star} gives
\begin{equation*}
 \E_\delta[\Pi_i]
 =
 \beta_i\left[
 d_i\left(\nu_i-\frac12\right)
 +\frac{\psi_i}{\zeta_{B_i}-\zeta_{W_i}}+S_i
 \right].
\end{equation*}

Once the fillers have conditional value zero, the separation between active nodes bounds every unwanted multiplication term. The auxiliary target therefore stays near $1/2$ for every choice of encoded values, before any source-gate equation is satisfied.

\begin{lemma}[Interior Range of the Auxiliary Target]\label{t:lem:Fmul-range}
If all fillers have $\theta=0$, then $15/32\le F_i^{\mathrm{mul}}\le17/32$ for every choice of the active players' conditional probabilities.
\end{lemma}

\begin{proof}
By \Cref{t:lem:sparse},
\[
 |h_i\overline R_i|
 \le
 \sum_{p\in\calA\setminus\{B_i,V_i\}}
 \left|\frac{h_i}{\zeta_{B_i}-\zeta_p}\right|
 \le\frac1{64}.
\]
Also $0\le h_i u_k\le h_i\le1/64$.  Thus
$|F_i^{\mathrm{mul}}-1/2| \le1/32$.
\end{proof}

\paragraph{Projecting the Gate Output.}

For the lower controller $L_i^F$, define the histogram table
\begin{equation*}
 \widehat F_i=
 \begin{cases}
 1-T^{(1)}_{L_i^F, X_j},&G_{1-},\\[1mm]
 2T^{(1)}_{L_i^F, X_j}-\dfrac14,&G_{\times2},\\[2mm]
 \dfrac9{32}
 +\dfrac18T^{(1)}_{L_i^F, X_j}
 +\dfrac18T^{(1)}_{L_i^F, X_k}
 +\dfrac12\Pi_i,&G_\phi.
 \end{cases}
\end{equation*}

The lower controllers implement the maximum with $1/4$; the upper controllers then implement the minimum with $3/4$.  The corresponding six residual tables are
\begin{align*}
d_{Z_i}
&=-1+4\left(
 T^{(1)}_{Z_i, L_i^0}
 +T^{(1)}_{Z_i, L_i^F}
 \right),\\
d_{L_i^0}
&=\frac14-T^{(1)}_{L_i^0, Z_i},\\
d_{L_i^F}
&=\calR(\widehat F_i)
 -T^{(1)}_{L_i^F, Z_i},\\
d_{X_i}
&=1-4\left(
 T^{(1)}_{X_i, U_i^z}
 +T^{(1)}_{X_i, U_i^b}
 \right),\\
d_{U_i^z}
&=T^{(1)}_{U_i^z, X_i}
 -\left(
 \frac72T^{(1)}_{U_i^z, Z_i}-\frac58
 \right),\\
d_{U_i^b}
&=T^{(1)}_{U_i^b, X_i}-\frac34.
\end{align*}
Here $\calR(\widehat F_i)$ means the pointwise affine transformation $2\widehat F_i/7+5/28$.

These equations specify every gate payoff. Each recovery statistic depends only on the observer's histogram, and its target is a distinct player. In particular, the observers $L_i^F,C_i,W_i$ are different from every logical player $X_j$, even for repeated inputs or source self-loops. The sum defining $d_{W_i}$ excludes $W_i$ itself, and its missing term is replaced by the constant $1/2$. Distinct nodes make all normalizers nonzero. We next prove the supports needed to recover the source equations.

\subsection{Correctness of the Reduction}\label{t:sec:gate-overview}

The recovery identities hold at the prescribed signal probabilities.
To apply them to every target equilibrium, we first use coarse calibration
and the gate inequalities to establish support. We can then compare the
two unperturbed actions, obtain exact calibration at zero error, and
derive an exact circuit solution. This order prevents the support argument
from assuming the gate equations it is meant to prove.

\paragraph{Support from Gate Inequalities.}

Every player's payoffs have the form
\begin{equation*}
 U_P(s,h)=G_P(K),\quad U_P(t,h)=0,\quad
 U_P(c,h)=G_P(K)+\eta d_P(h).
\end{equation*}
At prescribed signal probabilities write $D_P=\E_\delta[d_P]$. The gate argument requires only
\begin{equation}\label{t:eq:interval-complementarity}
 \theta_P<1\ \Longrightarrow\ D_P\le e,\qquad
 \theta_P>0\ \Longrightarrow\ D_P\ge-e.
\end{equation}
These one-sided inequalities are enough to establish the required supports and bound the gate error.

\begin{lemma}[Stability of the Gate Equations]\label{t:lem:semantic-stability}
Let $d_{\max}=\max(\{1\}\cup\{d_i:i\text{ bilinear}\})$. If~\cref{t:eq:interval-complementarity} holds with $0\le e\le1/[100(1+d_{\max})]$, every conditional probability $\theta_P$ is strictly below one and
\begin{equation}\label{t:eq:main-semantic-bound}
 |u_i-\mathcal P(f_i(u))|\le6(1+d_{\max})e.
\end{equation}
No source-gate equation is assumed.
\end{lemma}
\begin{proof}[Proof of \Cref{t:lem:semantic-stability}]
We follow the acyclic order in which support can be established: fillers, the cancellation pair, the two projection stages, and finally the input-copy pair.  All expectations in this proof are at the prescribed signal probabilities.

\emph{Fillers and the Cancellation Pair.}
A filler has residual $-1$, so its conditional probability is zero.  For the auxiliary pair $V_i,W_i$, the residuals are $\psi_i-1/2$ and $F_i^{\mathrm{mul}}-\nu_i$, where $F_i^{\mathrm{mul}}=1/2+h_i(u_k-\overline R_i)\in[15/32,17/32]$.  If $\nu_i=0$, the second residual forces $\psi_i=1$, and the first then forces $\nu_i=1$, a contradiction.  The case $\nu_i=1$ is symmetric.  Hence $0<\nu_i<1$, so $|\psi_i-1/2|\le e$.  This also puts $\psi_i$ in the interior and gives $|\nu_i-F_i^{\mathrm{mul}}|\le e$.

Write $P_i^\star=\E_\delta[\Pi_i]$ for the product table's expectation at the prescribed nodes.  The cancellation in \cref{t:eq:mult-cancellation-overview} now gives directly
\begin{equation}\label{t:eq:semantic-product-precopy}
 P_i^\star-\beta_i u_k
 =\beta_i\left[
 d_i(\nu_i-F_i^{\mathrm{mul}})
 +\frac{\psi_i-1/2}{\zeta_{B_i}-\zeta_{W_i}}
 \right].
\end{equation}
By \cref{t:eq:main-sparse}, its magnitude is at most $(65/64)d_i e<1/16$.  Notice that this estimate does not yet require $\beta_i$ to copy the first input.

\emph{The Two Projection Stages.}
Let $F_i^\delta$ be the ideal expectation of the raw gate table.  The preceding bound puts $F_i^\delta$ in $[-1/4,7/4]$ for every gate type.  Recall that $\mathcal R(z)=1/4+(2/7)(z-1/4)$ and $r_i=\mathcal R(z_i)$ is the lower-stage conditional probability.  Its two controllers have conditional probabilities $\lambda_i^0,\lambda_i^F$.  Their residuals, together with that of the value player, are
\[
 -1+4(\lambda_i^0+\lambda_i^F),
 \qquad \frac14-r_i,
 \qquad \mathcal R(F_i^\delta)-r_i.
\]
The two candidate values lie in $[3/28,19/28]$.  If $r_i=0$, the first controller is forced to one, contradicting the value-player inequality; if $r_i=1$, both controllers are forced to zero, giving the opposite contradiction.  Thus $0<r_i<1$, and
$|4(\lambda_i^0+\lambda_i^F)-1|\le e$.
Consequently, both controller probabilities are below one and at least one is positive.  Their inequalities place $r_i$ above both candidates within $e$ and below one candidate within $e$.  Therefore
$\left|z_i-\max\left\{\frac14,F_i^\delta\right\}\right| \le\frac72e$.

For the upper-stage value $u_i$, the controller probabilities $\mu_i^z,\mu_i^b$ have residuals
$1-4(\mu_i^z+\mu_i^b)$, $u_i-z_i$, $u_i-\frac34$.
Here $z_i\ge1/4-7e/2>0$.  If $u_i=0$, both controllers are forced to zero; if $u_i=1$, the boundary controller is forced to one.  Each case contradicts the value-player inequality.  As in the lower stage, $0<u_i<1$ forces both controller probabilities below one and at least one above zero.  Their inequalities give
\begin{equation}\label{t:eq:semantic-projection}
 \left|u_i-\min\left\{z_i,\frac34\right\}\right|\le e,
 \qquad
 |u_i-\mathcal P(F_i^\delta)|\le\frac92e.
\end{equation}
In particular, $1/4-9e/2\le u_i\le3/4+e$.

\emph{Copying the First Input.}
The remaining auxiliary pair has residuals $\chi_i-1/2$ and $u_j-\beta_i$.  The interval just proved gives $e<u_j<1-e$ and rules out both $\beta_i=0$ and $\beta_i=1$ by the same two-player contradiction used above.  Hence
$|\chi_i-1/2|\le e$, $|\beta_i-u_j|\le e$ and $0<\beta_i,\chi_i<1$.
Combining this with \cref{t:eq:semantic-product-precopy} gives
$|P_i^\star-u_j u_k|\le\left(1+\frac{65}{64}d_i\right)e$.
For a unary gate, $F_i^\delta=f_i(u)$.  For a bilinear gate their difference is $(P_i^\star-u_j u_k)/2$.  Since $\mathcal P$ is $1$-Lipschitz, \cref{t:eq:semantic-projection} yields \cref{t:eq:main-semantic-bound}.  The argument also gives $\theta_P<1$ for every player.
\end{proof}

\paragraph{From Equilibria to Circuit Solutions.}

We can now prove the quantitative reduction. Let $B$ bound the absolute values
of all gate payoff tables. The central roots and \Cref{t:lem:fekete-products,t:lem:reader-bits}
give $B\le2^{O(N)}$. Choose a rational $H\le2^{O(N)}$
large enough to bound $N,B,d_{\max}$ and the reciprocals of $b,c_*,\rho_0$.
One sufficient choice is $\eta=H^{-16}$ and $\xi_*=H^{-20}$.
The exact powers serve only to separate the error scales; \Cref{t:sec:compiler-correctness}
verifies that all the following hypotheses hold for $0\le\xi\le\xi_*$.

\begin{proof}[Proof of \Cref{t:thm:quantitative-reduction}]
Fix a $\xi$-$\WSNE$ of the unnormalized game. We first establish support.
The gate perturbation has magnitude at most $\gamma=\eta B$. Global
localization and the local inverse give
$\|w-\delta\|_\infty\le K_{\mathrm{cal}}(\eta B+\xi)<b/2$,
where $K_{\mathrm{cal}}=4N^2/(c_*b)$. Hence $0<w_P<1$ for every player.
Comparing actions $c$ and $s$, and then using \Cref{t:eq:main-reader-stability},
gives \Cref{t:eq:interval-complementarity} with
$e_0=\xi/\eta+2BNK_{\mathrm{cal}}(\eta B+\xi)$.
Our choice of scales makes $e_0\le1/[100(1+d_{\max})]$. By
\Cref{t:lem:semantic-stability}, every $\theta_P<1$. Since $0<w_P<1$,
each player therefore supports both $s$ and $t$.

We now compare these two unperturbed actions. Their expected payoff
difference is $|c_P|(1-\varrho_P(w))$, with no contribution from the gate
payoff. The local inverse consequently gives
\begin{equation}\label{t:eq:main-fine-w}
 \|w-\delta\|_\infty\le\frac{2N}{c_*}\xi.
\end{equation}
Transferring the gate expectations to $w=\delta$ once more gives
\Cref{t:eq:interval-complementarity} with the sharper error
$e(\xi)=\left(\eta^{-1}+\frac{4BN^2}{c_*}\right)\xi$.
The same stability lemma now yields
$|u_i-\mathcal P(f_i(u))|\le6(1+d_{\max})e(\xi)$.
This bound is proportional to the well-supported approximation error; the fixed perturbation $\eta B$ has
disappeared.

We recover a source value by
\begin{equation*}
 \widehat x_i=\clip_{[0,1]}\left(
 \frac{2\sigma_{X_i,c}}{\delta_{X_i}}-\frac12\right).
\end{equation*}
The projection is nonexpansive, and each source gate is $2$-Lipschitz
in the $\ell_\infty$ norm. Since $\sigma_{X_i,c}=w_{X_i}u_i$, replacing $w_{X_i}$
by $\delta_{X_i}$ changes the decoded coordinate by at most
$2\|w-\delta\|_\infty/b$. Thus every source residual is at most
$\Lambda\xi$ for a rational $\Lambda\le2^{O(N)}$.
\Cref{t:sec:compiler-correctness} gives the bound $\Lambda=H^{19}$.

At $\xi=0$, \Cref{t:eq:main-fine-w} gives $w=\delta$.
The second application of the stability lemma then places
$u_i\in[1/4,3/4]$ and makes every source gate residual zero.
Clipping is inactive, so
\Cref{t:eq:main-SL-decoder} is a separable-linear decoder.

Every unnormalized payoff lies in $[-3,3]$. The common transformation
$U\mapsto(U+3)/6$ puts all payoffs in $[0,1]$ and preserves exact
equilibria. For the normalized game we take
$\xi_{\mathrm{out}}=\xi_*/6$ and $\Lambda_{\mathrm{out}}=6\Lambda$.
The construction uses $O(G\log(2G))$ players and polynomially many
rational entries of polynomial description length, as proved in
\Cref{t:sec:compiler-correctness}. This completes the reduction.
\end{proof}

\section{Proofs of the Main Results}\label{sec:complexity}

The two reductions now yield the complexity classifications stated in the
introduction.

\paragraph{High-Precision Approximation.}

We use the population-padding argument of Chen, Durfee, and
Orfanou~\cite{cdo15} to pass from the core accuracy $m^{-54m^2}$ to
$2^{-n^c}$ for every fixed $c>0$.

\begin{proof}[Proof of \Cref{thm:two-strategy-main}]
Begin with the $m$-player game of \Cref{b:thm:core}, whose required accuracy is $\varepsilon_m=m^{-54m^2}$. Add players whose action $0$ pays $1/2$ and whose action $1$ pays zero. Extend each original player's payoff at the new counts by its existing constant value; for the two unrestricted players, use the value at count $m-1$. This preserves the constant action-$0$ payoff and introduces no new exceptional counts. At every well-supported equilibrium of accuracy below $1/2$, the added players choose $0$, so the original players face exactly their original count distributions.

Fix $c>0$. We choose a polynomially bounded population $n\ge m$ with
$32n\,2^{-n^c}\le\varepsilon_m^2$.
Such a choice exists because $\log(1/\varepsilon_m)=O(m^2\log m)$: take $n$ to be a sufficiently large fixed power of $m$. From a $2^{-n^c}$-$\ANE$, \Cref{lem:ane-wsne} gives an $\varepsilon_m$-$\WSNE$. Restricting to the original players gives an $\varepsilon_m$-$\ANE$ of the initial game. \Cref{b:thm:core,b:lem:source} then recover a \PureCircuit{} solution.

Membership follows from \Cref{sec:complexity-approximation}, using the stronger dyadic accuracy $2^{-n^K}$ for a fixed integer $K>c$ when necessary.
\end{proof}

The same lower bound holds with three strategies: we add a strictly
dominated strategy to the two-strategy game.

\begin{corollary}\label{cor:three-strategy-ppad}
For every fixed rational $c>0$, computing a $2^{-n^c}$-$\ANE$ in a three-strategy anonymous game is $\PPAD$-complete.
\end{corollary}
\begin{proof}
We add a third strategy whose payoff is always $0$ to the two-strategy game from \Cref{thm:two-strategy-main}. Strategy $0$ still pays $1/2$, so the added strategy is strictly dominated. We extend the original payoff functions by counting each opponent who plays the added strategy as playing strategy $1$. For sufficiently large $n$, \Cref{lem:ane-wsne} converts a $2^{-n^c}$-$\ANE$ of the three-strategy game into a $2^{-n^{c/2}}$-$\WSNE$, since $2^{-n^c}\le2^{-2n^{c/2}}/(48n)$. The added strategy receives no probability in this profile, so restricting to the original strategies gives a $2^{-n^{c/2}}$-$\WSNE$ of the two-strategy game. Apply \Cref{thm:two-strategy-main} with exponent $c/2$. A fixed amount of population padding handles the finitely many smaller populations. Membership follows from the result for any fixed number of strategies.
\end{proof}

\paragraph{Exact Equilibrium and Approximation in Distance.}

The decoder in \Cref{t:eq:main-SL-decoder} recovers each source coordinate
by a rational affine map of a single target probability.

\begin{proof}[Proof of \Cref{thm:three-strategy-main}]
By \Cref{t:thm:quantitative-reduction}, every exact target equilibrium decodes to an exact solution of \ExactGCircuit{}. Thus \Cref{t:thm:source} gives hardness under separable-linear reductions. Membership follows from \Cref{sec:complexity-approximation}.
\end{proof}

\begin{proof}[Proof of \Cref{cor:distance-main}]
The same decoder is $L$-Lipschitz for $L=\max_i\{1,2/\delta_{X_i}\}\le10$. A rational vector within $\varepsilon/L$ of a target equilibrium therefore decodes within $\varepsilon$ of an exact circuit solution. Clipping to $[0,1]$ cannot increase the distance, and the rescaled accuracy has polynomial encoding length. Membership follows from the Nash fixed-point map and the definition of $\FIXPa$~\cite{ey}.
\end{proof}

\section{Open Problems}\label{sec:discussion}

Exact computation with two strategies remains open. Our coefficient
construction gives only an $O(Nt)$ source-residual bound even at
an exact target equilibrium, for every positive scale $t$.
The current recovery leaves additional nonnegative products in the
selected coefficients. Decreasing their scales improves approximation
but does not establish an exact identity.
A $\FIXP$ reduction therefore needs an exact cancellation or recovery
identity, together with a global support argument. For a separable-linear
reduction, it must also replace the ratio decoder by a rational affine
function of one unconditional probability. Can one probability both
identify a player and encode a value?

The precision boundary is also open: for two-strategy anonymous games, algorithms achieve $O(n^{-1+\alpha})$ for every fixed $\alpha>0$, whereas our hardness result applies at $2^{-n^c}$ for every fixed $c>0$. Cheng, Diakonikolas, and Stewart~\cite{cds17} showed that a polynomial-time algorithm achieving $O(n^{-1-\alpha})$ for any fixed $\alpha>0$ would yield an FPTAS. Our exponential-precision hardness does not rule out such a scheme, whose running time may be polynomial in $1/\varepsilon$. At what intermediate accuracy does an aggregate distribution encode hard individual interactions?

For three strategies, removing the logarithmic overhead requires a different
multiplication or calibration argument: the sparse arrangement pays for a
harmonic sum of unwanted terms. Two coincident signal roots would isolate
their pair product exactly at second order, but the present global
calibration requires distinct nodes.

\paragraph{AI Methodology.}
OpenAI's Codex was utilized to assist with mathematical checks, literature research, and the drafting and revision of the exposition. The author takes full responsibility for the accuracy, originality, and integrity of all material presented.

\appendix
\section{Details of the Two-Strategy Reduction}\label{app:two-strategy-source}

We give the gate formulas for the quadratic source, followed by the coefficient
estimates used in the support and normalization arguments of
\Cref{sec:two-strategy-overview}. The final estimates verify the required precision.

\subsection{From \PureCircuit{} to Quadratic Equations}\label{b:app:source}

We replace each \PureCircuit{} gate with a short chain of quadratic
equations. The Boolean margins are checked against the gate's actual
inputs, so the argument also applies to circuits with cycles.

\begin{proof}[Proof of \cref{b:lem:source}]
For each constructed variable $x$, write $z=4x(1-x)\in[0,1]$.
We decode $z\le1/12$ as $0$, $z\ge11/12$ as $1$, and other values
as $\bot$.

\step{Boolean Amplification.}
A type L equation with $a=(1+\delta)/2$ and
$b=2(1-2\delta)$ implements $P(z)=1-[\delta+(1-2\delta)z]^2$
on the represented value $z$: its output is
$x=(1-\delta)/2-(1-2\delta)z/2$.
The map $P$ sends $[0,1]$ into itself, is decreasing, and has
Lipschitz constant at most $2$. Thus $P_2=P\circ P$ is increasing.
Evaluating the two fixed rational endpoints gives
$P^{10}(1/2)<1/16$ and $P^{10}(3/4)>63/64$.
Also $P(1/16)>63/64$ and $P(63/64)<1/16$.
Monotonicity therefore gives
\begin{align*}
 P^{10}\bigl([0,1/2]\bigr)&\subseteq[0,1/16],&
 P^{10}\bigl([3/4,1]\bigr)&\subseteq[63/64,1],\\*
 P^{11}\bigl([0,1/2]\bigr)&\subseteq[63/64,1],&
 P^{11}\bigl([3/4,1]\bigr)&\subseteq[0,1/16].
\end{align*}
These output intervals are separated from the decoding thresholds
by at least $1/48$.

\step{NOR and PURIFY Gadgets.}
Let $r(c)$ be the positive root of $r^2-r/10=c$, and define
$Q(c)=4r(c)(1-r(c))$. A type Q equation implements
$Q(c+\sum_j b_j z_j)$ by setting $d_j=4b_j$.
All arguments used below lie in $[1/4,0.89]$, where
$1/2<r(c)<1$ and $Q$ is decreasing. We use two thresholds.
Substituting $r=3/4$ shows that $c\le39/80$ implies $Q(c)\ge3/4$.
Also, $(6/7)^2-(6/7)/10=159/245<13/20$ and
$4(6/7)(1/7)=24/49<1/2$, so $c\ge13/20$ implies $Q(c)<1/2$.

We implement the three gate outputs by
\begin{align*}
 \NOR(z,z')&=P^{10}\bigl(Q(0.45+0.22(z+z'))\bigr),\\*
 \PURIFY_1(z)&=P^{11}\bigl(Q(0.45+0.4z)\bigr),\\*
 \PURIFY_2(z)&=P^{11}\bigl(Q(0.25+0.45z)\bigr).
\end{align*}
For NOR, two inputs decoding to $0$ give an argument at most
$73/150<39/80$, so the output decodes to $1$.
An input decoding to $1$ gives an argument at least
$391/600>13/20$, so the output decodes to $0$.

For PURIFY, if $z\ge1/2$, the first argument is at least $0.65$
and the first output decodes to $1$. If $z\le1/2$, the second
argument is at most $0.475$ and the second output decodes to $0$.
Thus at least one output is Boolean. If $z\le1/12$, both arguments
are at most $29/60<39/80$, so both outputs decode to $0$.
If $z\ge11/12$, both arguments are at least $53/80>13/20$,
so both decode to $1$.

\step{Variable Order and Approximation Error.}
Represent each wire by the last variable of its chain of
$10$ or $11$ applications of $P$. Put all type Q variables first,
followed by all type L chains, with the variables of each chain
in their forward order. Every type Q equation then uses later wire
outputs, while each type L equation uses the preceding variable in
its chain, beginning with the variable of the corresponding type Q
equation. This satisfies the required
order even when the original circuit has cycles.
There are at most $12$ variables per wire.

The three type Q constructions have coefficients
$(c;d_j)=(0.45;0.88,0.88)$, $(0.45;1.6)$, and $(0.25;1.8)$.
They satisfy \cref{b:eq:quadratic-margin}; in particular,
$c+\frac14\sum_jd_j\le0.89<0.9-\delta$.
Every type L equation has $a+b/4=1-\delta/2$,
as required by \cref{b:eq:linear-margin}.

It remains to check approximate solutions. If a type Q equation
has residual at most $h$, its variable differs from the positive
root by at most $3h$: the other root is below $-0.4$, while the
variable lies in $[0,1]$. Its represented $z$ therefore differs
by at most $12h$. A type L residual of $h$ changes the represented
output by at most $4h$.
We propagate errors in pairs of amplification steps.
Put $a=1-2\delta$ and $v=\delta+az$. Then
$(P_2)'(z)=4a^2v[\delta+a(1-v^2)]\le2a^2\le2$,
because $0\le v\le1$ and
$v(1-v^2)=v(1-v)(1+v)\le1/2$.
Each pair of approximate type L equations contributes at most $12h$,
so its error obeys $E_{k+1}\le2E_k+12h$, starting with $E_0\le12h$.
After five pairs, $E_5\le12(2^6-1)h$; a possible final step gives
$2E_5+4h<2^{11}h$.
This comparison holds with the gate's actual inputs fixed; it does
not propagate errors around a circuit cycle.
For $h\le\eps_{\mathrm{src}}$, the output error is at most
$2^{11}\eps_{\mathrm{src}}=1/64<1/48$, so the Boolean margins above
remain valid.
Decoding the variables representing the original wires therefore
solves \PureCircuit.
The construction has linear size and uses fixed rational coefficients.
\end{proof}

\subsection{Coefficient Estimates on Arbitrary Supports}\label{b:sec:cap}

We collect the coefficient-ratio and factorization estimates that
complement \Cref{b:lem:quarter-bound}, then absorb the equilibrium
error and prove the bound on arbitrary supports.
The first two estimates control coefficient ratios and locate large
weights; the third factors out a group whose weights dominate the
remaining total weight.

For $k\ge1$, define $\rho_k(\omega)=e_k(\omega)/e_{k-1}(\omega)$ when
$e_{k-1}(\omega)>0$, and set $\rho_k(\omega)=0$ otherwise.
The latter convention only concerns vanished coefficients: for a
nonnegative list, $e_{k-1}(\omega)=0$ implies $e_j(\omega)=0$ for all $j\ge k-1$.

\begin{lemma}[Elementary Symmetric Coefficient Bounds]\label{b:lem:coefficients}
Let $\omega=(\omega_1,\ldots,\omega_m)$ be a finite list of nonnegative weights.
\begin{enumerate}
\renewcommand{\labelenumi}{(\roman{enumi})}
\item \emph{Coefficient Ratios and Deletion.}
The sequence $\rho_1(\omega),\rho_2(\omega),\ldots$ is nonincreasing.
Deleting one entry places each ratio between two consecutive ratios
of the original list: for every $\ell\in[m]$ and $k\ge1$,
\begin{equation}\label{b:eq:interlace}
 \rho_{k+1}(\omega)\le\rho_k(\omega_{-\ell})\le\rho_k(\omega).
\end{equation}

\item \emph{Comparison with Ordered Weights.}
If $\omega_1\ge\cdots\ge \omega_m$, then, for $1\le k\le m$,
\begin{equation}\label{b:eq:sorted}
 \frac{\omega_k}{k}\le\rho_k(\omega)\le\sum_{j=k}^m \omega_j\le m \omega_k.
\end{equation}

\item \emph{Factoring Out a Designated Group.}
Suppose $\omega$ consists of positive designated weights $b_1,\ldots,b_h$
and a remaining list $J$. Write $P=\prod_{j=1}^h b_j$ and $H=e_1(J)$.
For an integer $k\ge0$, the products that include all $h$ designated
weights contribute exactly $Pe_k(J)$ to $e_{h+k}(\omega)$.
The full coefficient satisfies
\begin{equation}\label{b:eq:dominance}
 P e_k(J)\le e_{h+k}(\omega)
 \le P e_k(J)\prod_{j=1}^h\left(1+\frac{H}{b_j}\right).
\end{equation}
In particular, if $\theta=\sum_{j=1}^h H/b_j\le1/2$, then
$e_{h+k}(\omega)=Pe_k(J)(1+\theta_k)$ for some $0\le\theta_k\le2\theta$.

\end{enumerate}
\end{lemma}
We include the quantitative estimates used by the reduction.

The main application of \cref{b:eq:dominance} is $k=2$:
after factoring out $P$, the high-degree coefficient $e_{h+2}(\omega)$
approximates $e_2(J)$. If $J$ has just two nonzero entries $u,v$,
then $e_2(J)=uv$. Before we know that the relevant weights are
concentrated on two players, \cref{b:eq:third-order} provides the
corresponding upper bound with an explicit third-order error.

\begin{proof}[Proof of \cref{b:lem:coefficients}]
\step{Coefficient Ratios and Deletion.}
Newton's inequalities imply log concavity of the coefficients of
$\prod_j(1+\omega_jz)$; see Stanley~\cite{stanley89}. Applying
$e_k=a_k+\omega_\ell a_{k-1}$ to the coefficients $a_k$ of the list with entry $\omega_\ell$ deleted then gives
both deletion bounds in \Cref{b:eq:interlace}.

\step{Comparison with Ordered Weights.}
For the lower bound in \cref{b:eq:sorted}, extend each $(k-1)$-subset
by an absent member of the first $k$ entries. The sum of its possible
extension weights is at least $\omega_k$, while each resulting $k$-subset
is counted at most $k$ times. Hence $k e_k\ge \omega_k e_{k-1}$.
Every $k$-subset also has a member with index at least $k$, giving
$e_k\le e_{k-1}\sum_{j\ge k}\omega_j$ and the upper bound.

\step{Factoring Out a Designated Group.}
To prove \cref{b:eq:dominance}, expand according to the entries omitted
from the designated group:
\[
 e_{h+k}=P\sum_{j=0}^h
   e_j\!\left(\frac{1}{b_1},\ldots,\frac{1}{b_h}\right)e_{k+j}(J).
\]
The lower bound is the $j=0$ term. The inequality
$e_{k+j}(J)\le e_k(J)H^j$ gives the upper bound.
Writing $v=\sum_jH/b_j$, we have
$\prod_j(1+H/b_j)\le1/(1-v)\le1+2v$ when $v\le1/2$.

\end{proof}

\paragraph{A Two-Block Example.}
To see why the two equation types read in opposite directions,
consider the valid cyclic source system
\[
 x_1^2-\frac1{10}x_1=\frac9{20}+\frac{22}{25}y_2,
 \qquad
 x_2=\frac{1-\delta}{2}-2(1-2\delta)y_1,
 \qquad y_j=x_j(1-x_j).
\]
Here $q=2$, $N=11$, $r_1=2$, and $r_2=6$.
For this illustrative calculation, omit lower-scale terms and take
the prefix and terminal products to be normalized.
The coefficient used by the type L row of $U_2$ selects the five
players $B,1,C_1,U_1,V_1$, giving
$\frac{e_5^{-U_2}}{\beta S_4}\approx u_1v_1$.
Deleting $U_2$ leaves that earlier prefix intact.
By contrast, the type Q row of $U_1$ reads the prefix through $V_2$,
which contains nine players including $B$. Deleting $U_1$ divides
its product by $s_1u_1$ and leaves eight factors, so
\[
 \frac{s_1e_8^{-U_1}}{\beta S_8}\approx\frac{u_2v_2}{u_1},
 \qquad
 \frac{s_1e_{10}^{-U_1}}{\beta S_{10}}\approx\frac1{u_1}.
\]
The second term supplies the type Q constant. Multiplying the
$U_1$ payoff condition by $u_1$ gives
$u_1(v_1-9/10)+9/20+(22/25)u_2v_2\approx0$.
The $U_2$ condition gives
$v_2-(1+\delta)/2-2(1-2\delta)u_1v_1\approx0$.
Substituting $v_j=1-u_j$ recovers the two source equations.
This calculation explains the indices and the reciprocal factor;
\cref{b:sec:correctness} proves the coefficient estimates with their errors,
after mixing and normalization have been established.

\paragraph{Absorbing the Equilibrium Error.}

This absorption requires no lower bound on an individual probability.
Before absorption, each row gives $e_{k+h}\le a e_k+32\zeta$
for its target $a>0$ and integers $k,h\ge1$. This implies
$e_{k+h}\le a\left(1+\frac{32\zeta}{a^{1+k/h}}\right)e_k$.
Indeed, if $e_k=0$, all higher coefficients vanish. If
$e_{k+h}\le a e_k$, the bound is immediate. Otherwise the
coefficients through $e_{k+h}$ are positive. Decreasing ratios give
$e_{k+h}/e_k\le\rho_{k+1}^h$, hence $\rho_{k+1}>a^{1/h}$.
They also give $e_k=\prod_{j=1}^k\rho_j>a^{k/h}$.
Thus $a e_k>a^{1+k/h}$, and dividing the additive error by
$a e_k$ proves the estimate.
For the rows of player $1$, $C_i$, $V_i$, and $Z_i$, respectively,
the exponents of $t$ in this last lower bound are
$2$, $(7i-2)(4i-1)$, $7i(4i+1)$ and $(7i-1)(4i+2)$.
Their maximum is $28q^2+10q-2=2d_q-(12q+4)$.
Since $\zeta=S^2=t^{2d_q}/49$, the additive error divided by
$a e_k$ is at most $(32/49)t^{12q+4}\le t^{10}$ for $q\ge1$.
This proves the primed upper bounds, including at zero weights.

\paragraph{Bounding Coefficients before Mixing.}

Recall the global coefficients $M_i,R$ from
\cref{b:eq:global-coefficients}. We prove the bounds in \cref{b:lem:cap}
using only the support restrictions and control inequalities
already established; the intended pair need not mix.

\begin{proof}[Proof of \cref{b:lem:cap}]
\step{Initializing the Prefix Coefficients.}
For this proof, introduce the auxiliary quantities
\[
 A_i=\frac{E_{r_i+1}}{\beta S_{r_i}},\qquad
 B_i=\frac{E_{r_i}}{\beta S_{r_i-1}},\qquad A_{q+1}=R.
\]
We first bound these quantities by the same expression for every block.
For the first block, let $H$ be the total weight other than that
of $B$ and player $1$. By \cref{b:lem:auxiliary}, $\omega_B\ge1$ and
$p_1>0$. The two inequalities for player $1$ imply
$\omega_BH\le t(1+\xi)(\omega_B+H)$ and $t\le H$.
Here the latter follows from $e_2^{-1}/e_1^{-1}\le H$.
Hence
\begin{equation}\label{b:eq:initial-weights}
 t\le H\le t(1+6t),\qquad
 \left|\omega_1-\frac17\right|\le4t,\qquad
 |\omega_B-\beta|\le80t.
\end{equation}
For the first bound, solve the upper inequality for $H$ using $\omega_B\ge1$.
For the others, the total probability outside $B$ and player $1$ is
at most $H\le2t$, and $\eta\le t^2$.
The auxiliary equations therefore give $|p_1-T|\le3t$ and
$|b-b_0|\le3t$. On the relevant intervals, the map $p\mapsto p/(1-p)$
has derivative at most $4/3$ for $p_1$ and $16$ for $b$.
These bounds imply \cref{b:eq:initial-weights}.
Expanding $E_2$ and $E_3$ according to the entries $\omega_B,\omega_1$
and using $e_2\le H^2/2$ and $e_3\le H^3/6$ for the remaining list
gives
$A_1,B_1\le(1+48t)(1+28t)(1+6t)+10t\le1+100t$.

\step{Using the Product Row before Mixing.}
Fix a block and suppress its subscript. By deletion interlacing,
the $C$ and $V$ rows imply the full-list inequalities
\begin{equation}\label{b:eq:full-control-bounds}
 E_{r+2}\le s'E_{r+1},\qquad E_{r+4}\le\tau'E_{r+3}.
\end{equation}
Applying these bounds directly gives
$E_{r+5}\le(s')^2(\tau')^2E_{r+1}$.
The reference prefix product is multiplied by $\tau K$, so the multiplier
above is a factor of order $1/t$ too large to propagate the normalized
prefix bound. We use the $Z$ constraint to obtain the sharper estimates
below.
Put $\Delta=s'(\tau')^2/K=t(1+\xi)^3$ and
$\gamma=(1-\Delta)^{-1}$. We claim that
\begin{equation}\label{b:eq:propagation}
 \begin{aligned}
 E_{r+4}&\le\bigl(K'+\gamma\tau'(s')^2\bigr)E_{r+1},\\
 E_{r+5}&\le\gamma\tau'K'E_{r+1}.
 \end{aligned}
\end{equation}
If $\omega_Z=0$, its upper constraint gives $E_{r+4}\le K'E_{r+1}$,
while decreasing ratios and \cref{b:eq:full-control-bounds} give
$E_{r+5}\le\tau'E_{r+4}$. Thus the claim holds.

Suppose $\omega_Z>0$. We use the $Z$ lower constraint to force an
intermediate coefficient ratio to be large. The $V$ upper constraint
then bounds $\omega_Z$ at the scale $\tau$, without assuming that $V$ mixes.
By \cref{b:eq:D}, the strict lower constraint for $Z$ gives
$e_{r+4}^{-Z}>K e_{r+1}^{-Z}\ge0$.
The list excluding $Z$ therefore has at least $r+4$ positive entries;
deleting $V$ leaves at least $r+3$. All denominators below are positive.
Set
\[
 g=\frac{e_{r+2}^{-V,-Z}}{e_{r+1}^{-V,-Z}},\qquad
 h=\frac{e_{r+3}^{-V,-Z}}{e_{r+2}^{-V,-Z}}.
\]
The $C$ and $V$ upper constraints, together with deletion
interlacing, give the following inequalities even if either player
has weight zero:
\[
 \frac{e_{r+2}^{-Z}}{e_{r+1}^{-Z}}\le s',\qquad
 \frac{e_{r+4}^{-Z}}{e_{r+3}^{-Z}}\le\tau'.
\]
Combining these bounds with the $Z$ lower constraint yields
$g\ge\frac{e_{r+3}^{-Z}}{e_{r+2}^{-Z}}\ge\frac{K}{s'\tau'}$.
For the second inequality, factor
$e_{r+4}^{-Z}/e_{r+1}^{-Z}$ into its three consecutive ratios:
the first is at most $s'$, the last at most $\tau'$, and their
product with the middle ratio is greater than $K$.
The first inequality is deletion interlacing for $V$.
Thus no lower payoff constraint for $C$ or $V$ has been used.
To use the $V$ upper constraint, expand its two coefficients in the
entry $\omega_Z$:
\begin{align*}
 e_{r+3}^{-V}&=e_{r+2}^{-V,-Z}(h+\omega_Z),\\
 e_{r+2}^{-V}&=e_{r+2}^{-V,-Z}\left(1+\frac{\omega_Z}{g}\right).
\end{align*}
Dividing by $e_{r+2}^{-V,-Z}>0$ gives
$h+\omega_Z\le\tau'(1+\omega_Z/g)$.
The lower bound on $g$ gives $\tau'/g\le\Delta$, so
$h+\omega_Z\le\tau'+\Delta \omega_Z$.
Since $0\le \omega_Z\le h+\omega_Z$, we obtain $h+\omega_Z\le\gamma\tau'$.
Deletion interlacing also gives
$e_{r+5}^{-Z}\le h e_{r+4}^{-Z}$. Hence
$E_{r+5}\le(h+\omega_Z)e_{r+4}^{-Z}\le\gamma\tau'K'E_{r+1}$.
Finally, \cref{b:eq:full-control-bounds} and decreasing ratios give
$E_{r+3}\le(s')^2E_{r+1}$, and therefore
$E_{r+4}\le K'E_{r+1}+\omega_Z E_{r+3} \le\bigl(K'+\gamma\tau'(s')^2\bigr)E_{r+1}$.
This proves \cref{b:eq:propagation} for every support.
This first use of the product row controls coefficient growth;
it does not identify any coefficient with an intended pair's product.

Since successive reference prefix products differ by the factor
$s_i^2\tau_i\chi_{i+1}$, \cref{b:eq:propagation} gives
\[
 \begin{aligned}
 A_{i+1}&\le\gamma(1+\xi)^2A_i\le(1+2t)A_i,\\
 B_{i+1}&\le\bigl[t(1+\xi)+\gamma(1+\xi)^3\bigr]A_i
          \le(1+3t)A_i.
 \end{aligned}
\]
The second recurrence is used for $i<q$.
Since $N\ge7$ and $Nt\le2^{-10}$, iteration gives
\begin{equation}\label{b:eq:prefix-coefficient-cap}
 A_i,B_i,R\le\exp(103t+2Nt)
             \le\exp(18Nt)\le1+64Nt\le2.
\end{equation}
Here we used $e^x\le1+2x$ for $0\le x\le1$.

\step{Obtaining the Quarter Bound from the Third Coefficient.}
We factor out the largest weights and use the two full-list
constraints to control the remaining second coefficient.
The $C_i$ upper inequality and decreasing ratios give
$M_i\le A_i(1+\xi)^2$, which already gives $M_i\le1/4$ when
$A_i\le1/8$. We use a fixed cutoff: a smaller prefix coefficient
needs no factorization, while a larger one gives a lower bound
on each large weight. Suppose $A_i>1/8$ and put $r=r_i$.
Then $E_{r+1}>0$ implies $E_r>0$, so $B_i$ and the following
coefficient ratios are positive.
Using these coefficient bounds and \cref{b:eq:interlace}, we have
$\rho_{r+1}=(A_i/B_i)\chi_i\ge\chi_i/16$ and $\rho_{r+2}\le s_i'$.
By \cref{b:eq:sorted}, the largest $r+1$ weights are at least
$m=\chi_i/(16N)$, and the remaining weights have total weight at most
$2N^2s_i$. Call these two groups the large entries and the remaining
list $J$, of total weight $H$, and let $P$ be the product of the large
entries. The sum of the ratios $H/b_j$ in \cref{b:eq:dominance} is at
most $N(2N^2s_i)/m=32N^4t^4$.
By $Nt\le2^{-10}$, this is at most $1/2$, so
\cref{b:eq:dominance} applies with
$\theta=64(Nt)^4\le Nt/128$ and gives
$P e_k(J)\le E_{r+1+k}\le P e_k(J)(1+\theta)$ for each $k\ge0$.
The full-list inequalities in \cref{b:eq:full-control-bounds} give
\[
 \begin{aligned}
 PH&\le E_{r+2}\le s_i'E_{r+1}\le s_i'P(1+\theta),\\
 Pe_3(J)&\le E_{r+4}\le\tau_i'E_{r+3}
          \le\tau_i'Pe_2(J)(1+\theta).
 \end{aligned}
\]
Consequently, $H\le s_i'(1+\theta)$ and
$e_3(J)\le\tau_i'(1+\theta)e_2(J)$, regardless of the positions
of $C_i$ and $V_i$ in the weight ordering.
Applying \Cref{b:lem:quarter-bound} directly gives
$e_2(J)\le H^2/4+(3/4)\tau_i'(1+\theta)H$,
including when coefficients vanish.

Finally, $P\le E_{r_i+1}$. Applying \cref{b:eq:dominance} to
$E_{r_i+3}$ and using the bound on $H$ yields
\[
 \begin{aligned}
 M_i&\le\frac{P e_2(J)(1+\theta)}{\beta S_{r_i}s_i^2}
       \le A_i(1+\theta)\frac{e_2(J)}{s_i^2}\\
 &\le A_i(1+\xi)^2(1+\theta)^3
       \left(\frac14+\frac34t^2\right)\\
 &\le\frac14\exp(32Nt)
 \le\frac14+16Nt.
 \end{aligned}
\]
The exponential bound uses $\xi,\theta,t^2\le Nt$ and
\cref{b:eq:prefix-coefficient-cap}; also $32Nt<1$.
\end{proof}

The bound $1/4$ is needed to use the source margins before mixing
has been proved. For a type Q equation,
\cref{b:eq:quadratic-margin} gives
$c_i+\tfrac14\sum_jd_{ij}\le a_i-\delta$.
The bounds on $M_j$ and $R$ therefore imply
$c_iR+\sum_jd_{ij}M_j\le a_i-\delta+O(Nt)$, leaving a fixed gap
below $a_i$. This gap will rule out a small value of $v_i$ in
\cref{b:eq:force-v}. A bound of $1/2$ for $M_j$ would not suffice:
the same expression could exceed $a_i$, and that argument would lose
its gap. Since \cref{b:lem:cap} holds on arbitrary supports, we may use
it to force mixing without first identifying $M_j$ with the intended
pair's product.

\subsection{Mixing, Compensation, and Precision}\label{b:sec:induction}

Once the coefficient bound is available, we can force mixing within a
block and pass its bounds to the next. We then verify that a single
precision choice suffices for the whole reduction.

\paragraph{The Block Step.}

We prove \cref{b:lem:block-step}.
The constants in \cref{b:eq:block-constants} satisfy
$\Lambda=2^{14}$ and $\rho=2^{-28}$.
The block and decoding proofs use the following common smallness
bounds, all consequences of $Nt\le\rho$:
\begin{equation}\label{b:eq:common-smallness}
 \begin{gathered}
 \Lambda Nt\le\frac{\alpha}{8},\qquad
 64Nt+4t^2\le\frac{\delta}{8},\qquad t<\frac{\kappa}{8},\\
 800t^2\le Nt,\qquad 770Nt^3\le Nt,\qquad
 \xi,\nu\le Nt^3.
 \end{gathered}
\end{equation}
Indeed, $t,Nt\le\rho$, $\Lambda\rho=\alpha/8$, and
$64Nt+4t^2\le68\rho<\delta/8$.
Also $\rho<\kappa/8$, $800\rho<1$, and $770\rho^2<1$,
which give the bounds on $t,t^2,Nt^3$. Finally, $\xi=t^{10}\le t^3$
and $\nu\le6S^2\le t^3$, since $d_q\ge26$ and $S=t^{d_q}/7$.
The same threshold implies the hypothesis of \cref{b:lem:cap}
and $Nt^2\le1$. The following estimates use these bounds.

\begin{proof}[Proof of \cref{b:lem:block-step}]
\step{Forcing the Variable Pair to Mix.}
Throughout this step, \cref{b:lem:cap} is already available for every
block, including the current and later ones. Its proof did not use
the block entry hypotheses or any variable-pair mixing conclusion.
The coarse lower bounds and \cref{b:eq:scales} imply that the prefix
through $C_i$, together with $B$, has minimum weight at least
$\chi_i/16$. Also $H_i\le2s_i$ and $1/2\le\Pi_i\le2$.
The sum of the tail-to-prefix ratios in \cref{b:eq:dominance}
is at most $32Nt^4\le32Nt^3$. Hence every factorization using this prefix,
including after deleting $U_i,V_i$, or $Z_i$, has relative error
between $0$ and $\varepsilon_{\mathrm{tail}}=64Nt^3$.
In particular, the two local terms in the $U_i$ row equal
\[
 \Pi_i\left(\frac{H_i}{s_i}-u_i-a_i\right)+\varepsilon_i,
 \qquad |\varepsilon_i|\le384Nt^3.
\]
Indeed, the factorization errors are bounded by
$\Pi_i[(H_i/s_i-u_i)+a_i]\varepsilon_{\mathrm{tail}}\le6\varepsilon_{\mathrm{tail}}$.

For type L, each input is at most
$\frac{e_{r_j+3}^{-U_i}}{\beta S_{r_j+2}} \le M_j\le\frac14+16Nt$,
by \cref{b:lem:cap}. For type Q, every input and constant term is
nonnegative. The one-sided inequality $\D_{U_i}\le\nu$ therefore gives
\[
 u_i\ge
 \begin{cases}
 1-a_i-\frac14\sum_{j<i}b_{ij}-2\alpha-64Nt,
       &\text{type L},\\
 1-a_i-\alpha-Nt,&\text{type Q}.
 \end{cases}
\]
For type L, we used $\Pi_i\ge1-\alpha$ and
$\frac14\sum_jb_{ij}<1/2$ to obtain
$\frac{1}{4\Pi_i}\sum_jb_{ij}\le\frac14\sum_jb_{ij}+\alpha$.
The remaining type L error is at most
$32Nt/(1-\alpha)+770Nt^3\le64Nt$; for type Q it is at most
$Nt$, by \cref{b:eq:common-smallness}.
The source margins and \cref{b:eq:common-smallness}
make these lower bounds at least $\delta/4$ and $1/20$,
respectively. Thus $u_i\ge\kappa$ and $U_i$ mixes.
Only its one-sided inequality was used.

The upper constraint for $V_i$ now bounds the weight after the pair.
After deleting $V_i$, the list following the prefix is
$(\omega_{U_i},\omega_{Z_i},J_i)$, whose second coefficient is at least
$\omega_{U_i}T_i$. Put $\vartheta=(1+\xi)(1+\varepsilon_{\mathrm{tail}})$.
Since $\xi\le Nt^3$, we have $\vartheta\le1+192Nt^3\le2$, and
\[
 \omega_{U_i}T_i\le\tau_i\vartheta(\omega_{U_i}+T_i),\qquad
 \frac{T_i}{\tau_i}
 \le\frac{\vartheta}{1-\frac{\vartheta\tau_i}{\omega_{U_i}}}
 \le1+9\kappa^{-1}t^2.
\]
Here $\omega_{U_i}\ge\kappa s_i$ and
$\vartheta\tau_i/\omega_{U_i}\le2\kappa^{-1}t^2<1/2$.
In particular, $T_i\le2\tau_i$.
Absorbing this tail into the local expression gives
\begin{equation}\label{b:eq:local-row}
 \frac{e_{r_i+2}^{-U_i}}{\beta S_{r_i}s_i}
 -\frac{a_i e_{r_i+1}^{-U_i}}{\beta S_{r_i}}
 =\Pi_i(v_i-a_i)+\lambda_i,\qquad
 |\lambda_i|\le400t^2.
\end{equation}
The additional error is at most $2T_i/s_i\le4t^2$;
here $Nt\le1$ absorbs the preceding $384Nt^3$ error.

Since $U_i$ mixes, $\D_{U_i}>0$.
For type L, the inputs are subtracted, so
$v_i\ge a_i-800t^2>\kappa$.
For type Q, multiply this positivity by $u_i\le2$.
The exact inequalities
$\omega_{U_i}e_{r_j+2}^{-U_i}\le E_{r_j+3}$ and
$\omega_{U_i}e_{N-1}^{-U_i}=E_N$ imply
\begin{equation}\label{b:eq:force-v}
 \Pi_i u_i(a_i-v_i)
 \le c_iR+\sum_{j>i}d_{ij}M_j+800t^2
 \le a_i-\delta+64Nt.
\end{equation}
The last bound uses the source margin and
$64c_i+16\sum_jd_{ij}\le64a_i<63$,
together with \cref{b:eq:common-smallness}.
If $v_i<\kappa$, then
$u_i\ge1-\alpha-\kappa-2t^2$. Since $\Pi_i\ge1-\alpha$ and
$0<a_i\le1$, the left-hand side satisfies
\[
 \Pi_i u_i(a_i-v_i)
 \ge(1-\alpha)(1-\alpha-\kappa-2t^2)(a_i-\kappa)
 \ge a_i-2\alpha-2\kappa-4t^2.
\]
All three factors in the first product are positive by
\cref{b:eq:common-smallness}. Now $2\alpha+2\kappa=5\delta/8$, and
$4t^2+64Nt\le\delta/8$ by \cref{b:eq:common-smallness}.
This contradicts \cref{b:eq:force-v}.
Thus $v_i\ge\kappa$ and $V_i$ mixes in both cases.

Its lower constraint also gives $T_i\ge\tau_i/(1+\varepsilon_{\mathrm{tail}})$:
the second-to-first coefficient ratio of
$(\omega_{U_i},\omega_{Z_i},J_i)$ is at most $T_i$.
Together with the upper bound, this proves
$|T_i/\tau_i-1|\le9\kappa^{-1}t^2$.

\step{Forcing the Product Constraint.}
The total-weight normalization gives a stronger product bound than
the separate lower bounds alone. Indeed,
$h_i:=u_i+v_i=\frac{H_i}{s_i}-\frac{T_i}{s_i} \ge1-\alpha-2t^2\ge\frac{3}{4}$.
Since $u_i,v_i\ge\kappa$, we have
$u_iv_i\ge\kappa(h_i-\kappa)\ge\kappa/2$.
This bound keeps the compensation and propagation constants linear
in $\kappa^{-1}$.
If $p_{Z_i}=0$, then $e_1(J_i)=T_i\ge\tau_i/2$.
Its upper constraint would give
$e_3(\omega_{U_i},\omega_{V_i},J_i)\le2K_i$, whereas
\[
 e_3(\omega_{U_i},\omega_{V_i},J_i)
 \ge \omega_{U_i}\omega_{V_i}e_1(J_i)
 \ge\frac{\kappa s_i^2\tau_i}{4}>2K_i.
\]
The last inequality uses $K_i=s_i^2\tau_i t$ and
\cref{b:eq:common-smallness}.
Thus $Z_i$ mixes. Its two-sided constraint and the same prefix
factorization give
$\left|\frac{e_3(\omega_{U_i},\omega_{V_i},J_i)}{K_i}-1\right|\le192Nt^3$.
Expand the coefficient as
\[
 e_3(\omega_{U_i},\omega_{V_i},J_i)
 =\omega_{U_i}\omega_{V_i}e_1(J_i)
 +(\omega_{U_i}+\omega_{V_i})e_2(J_i)+e_3(J_i).
\]
Since $\omega_{U_i}+\omega_{V_i}\le2s_i$ and $e_1(J_i)\le2\tau_i$,
the last two terms sum to at most $6s_i\tau_i^2=6tK_i$.
Since $Nt^2\le1$, the total product error is at most $200t$.
Consequently,
\begin{equation*}
 \begin{aligned}
 \left|\frac{\omega_{U_i}\omega_{V_i}e_1(J_i)}{K_i}-1\right|
       &\le200t,\\
 e_1(J_i)&=\frac{\chi_{i+1}}{u_iv_i}(1+\theta_i),
       \qquad |\theta_i|\le200t.
 \end{aligned}
\end{equation*}
The bounds $u_i+v_i\le2$ and $u_iv_i\ge\kappa/2$ imply
$\chi_{i+1}/2\le e_1(J_i)\le4\kappa^{-1}\chi_{i+1}$.
Using $\omega_{Z_i}=T_i-e_1(J_i)$ and $t\le1$, we obtain
\[
 \left|\frac{\omega_{Z_i}}{\tau_i}-1\right|
 \le9\kappa^{-1}t^2+4\kappa^{-1}t
 \le13\kappa^{-1}t\le\Lambda t.
\]
In particular, $\tau_i/2\le \omega_{Z_i}\le2\tau_i$.

\step{Bounding the Next Normalization and Prefix Increment.}
Suppose $i<q$.
The prefix through $Z_i$, together with $B$, has minimum weight
at least $\tau_i/2$, and the remaining total weight is at most
$4\kappa^{-1}\chi_{i+1}$.
Thus \cref{b:eq:dominance} applies with relative error at most
$\Omega=16\kappa^{-1}Nt$, also after deleting $C_{i+1}$.
If $p_{C_{i+1}}=0$, its coefficient ratio would be at least
$\chi_{i+1}/[2(1+\Omega)]\ge\chi_{i+1}/4>s_{i+1}'$.
This contradicts its upper constraint.
Hence $C_{i+1}$ mixes, and its two-sided constraint gives
$\frac{1}{1+\Omega} \le\frac{H_{i+1}}{s_{i+1}} \le(1+\xi)(1+\Omega)$.
This is a fresh estimate from the $C_{i+1}$ row. Its error comes from
factoring the current prefix; the preceding normalization is used
only to establish coarse weight bounds.
In particular,
\begin{equation*}
 \left|\frac{H_{i+1}}{s_{i+1}}-1\right|\le32\kappa^{-1}Nt,\qquad
 \left|\frac{\omega_{C_{i+1}}}{e_1(J_i)}-1\right|\le4t^3.
\end{equation*}
The first uses $\xi\le\Omega/2$ and $\Omega\le1/2$, so the
deviation is at most $2\Omega$. The second uses
$\omega_{C_{i+1}}=e_1(J_i)-H_{i+1}$,
$H_{i+1}\le2s_{i+1}$, and $e_1(J_i)\ge\chi_{i+1}/2$.
The upper bound $\omega_{C_{i+1}}\le4\kappa^{-1}\chi_{i+1}$ follows
from $\omega_{C_{i+1}}\le e_1(J_i)$. For the lower bound,
$\omega_{C_{i+1}}\ge\chi_{i+1}(1/2-2t^3)\ge\chi_{i+1}/16$.
Thus both coarse bounds needed by the next block hold.

The exact prefix identity is
\begin{equation*}
 \frac{\Pi_{i+1}}{\Pi_i}
 =\frac{\omega_{Z_i}}{\tau_i}
  \frac{\omega_{U_i}\omega_{V_i}e_1(J_i)}{K_i}
  \frac{\omega_{C_{i+1}}}{e_1(J_i)}.
\end{equation*}
The respective errors are bounded by $13\kappa^{-1}t$,
$200t$, and $4t^3$.
The last factor is at most one, and the middle factor is at most $3/2$.
Therefore,
$\left|\frac{\Pi_{i+1}}{\Pi_i}-1\right| \le(26\kappa^{-1}+204)t\le\Lambda t$.
For the last block, $J_q=(\omega_A)$, so $e_1(J_q)=\omega_A$ exactly.
Player $A$ mixes by \cref{b:lem:auxiliary}.
The corresponding identity contains only the first two factors:
$\frac{R}{\Pi_q} =\frac{\omega_{Z_q}}{\tau_q} \frac{\omega_{U_q}\omega_{V_q}e_1(J_q)}{K_q}$.
The same bounds give $|R/\Pi_q-1|\le\Lambda t$.
Every estimate used only the coarse hypotheses and
\cref{b:lem:cap}; the constants are the same for all blocks.
\end{proof}

\paragraph{Precision and Rational Payoffs.}

To obtain unscaled tolerance $\eta=\zeta/8$, apply
\Cref{lem:ane-wsne} at normalized well-supported tolerance
$\zeta/(16L)=S^3/128$. The required ANE accuracy is
$(S^3/128)^2/(32N)=S^6/(2^{19}N)=\varepsilon_0$.

The constants used above are $\Lambda=2^{14}$, $\rho=2^{-28}$,
and $A_{\mathrm{dec}}=2^{19}$.
For $t=2^{-36}/N$ and every $N\ge7$, we have
$Nt=2^{-36}<\rho$ and
$A_{\mathrm{dec}}Nt=2^{-17}=\eps_{\mathrm{src}}$.
They also imply $t\le2^{-10}$, so the preceding support and tolerance
estimates hold. Thus \cref{b:lem:cap,b:lem:block-step,b:lem:normalization}
apply, and \cref{b:eq:residual-L,b:eq:residual-Q} meet the source tolerance.
The decoded vector therefore solves the given instance of \OQE.
For the hard instances produced by \cref{b:lem:source}, its Boolean
decoding also recovers a valid \PureCircuit{} assignment.

Set $C=54$. Here $d_q=14q^2+11q+1\le7N^2/8$.
For $N\ge16$, the bound $2^{36}\le N^9$ gives $t^{-1}\le N^{10}$,
and hence $\eps_0^{-1}=2^{19}N7^6t^{-6d_q}
\le N^{26+60d_q}\le N^{53N^2}$.
For the remaining populations $N\in\{7,11,15\}$, direct substitution gives
$\eps_0^{-1}\le N^{54N^2}$ as well,
so $\eps_N=N^{-CN^2}\le\eps_0$.
Thus every $\eps_N$-ANE is sufficiently accurate for the reduction.

The hard source has linear size in the \PureCircuit{} instance, and the game
has $N=4q+3$ players and $O(N^2)$ payoff entries. The polynomial-time
conversion in \Cref{lem:ane-wsne} produces rational probabilities of
polynomial encoding length. Evaluating their weights and the ratios defining
$x_i$ also takes polynomial time; Boolean decoding is immediate.
Thus the construction and decoder form a polynomial-time reduction,
and the support proof covers every target approximate equilibrium.

Every listed count is at most $N-1$. The exponents in individual
reference weights are $O(N)$, and those in reference products are
$O(N^2)$. More precisely,
$S=\frac17t^{d_q}$ and $d_q=14q^2+11q+1$.
Apart from the fixed factor $\beta$, each denominator in
\cref{b:tab:payoffs} is a product of a subset of the reference weights,
and hence is at least $S$.
The source coefficient bounds give total absolute corrections
of at most $6/(\beta S)$ in any encoding entry. Therefore
$|f_\ell(k)|\le\zeta+\frac{6}{\beta S}<\frac8S=L$.
The auxiliary differences have absolute value at most $N<L$,
since $S\le t/7<1/N$. Thus all payoffs lie in $[0,1]$.
If each source coefficient has at most $B$ bits, then
$\log(1/t)=O(\log N)$ implies that each rational payoff entry has
$O(B+N^2\log N)$ bits. In the hard instances produced by
\cref{b:lem:source}, $B=O(1)$.
The whole game can therefore be constructed in polynomial time.

\section{Details of the Three-Strategy Reduction}\label{app:three-strategy-details}

We construct the rational interpolation nodes, prove the Vandermonde gap,
and verify the separation and error estimates used in
\Cref{sec:three-strategy-overview}. These arguments supplement the complete gate
construction and support proof in the main text.

\subsection{Calibration Nodes and Global Localization}\label{t:sec:fekete}

We begin with the one-dimensional Fekete facts used by both calibration and the sparse multiplication arrangement. On an interval, these maximizers are the shifted Legendre--Gauss--Lobatto nodes: the endpoints together with the zeros of the derivative of a Legendre polynomial. We derive their central-spacing bounds from the differential equation for the corresponding Jacobi polynomial, and then quantify the Vandermonde gap that excludes distant signal profiles.

\paragraph{Fekete Points on an Interval.}

For ordered $0\le y_1<\cdots<y_N\le1$, write
\[
 V(y)=\prod_{1\le i<j\le N}(y_j-y_i).
\]

\begin{lemma}[Fekete Characterization]\label{t:lem:fekete-characterization}
The unique ordered maximizer
$0=x_1<x_2<\cdots<x_N=1$
of $V$ on $[0,1]$ consists of the endpoints together with the shifted zeros of $P_{N-1}'$, where $P_{N-1}$ is the Legendre polynomial of degree $N-1$ on $[-1,1]$.
\end{lemma}

This is the classical characterization of Fej\'er~\cite{fejer32}; we use it directly. Standard properties of the corresponding Jacobi polynomials are collected by Szeg\H{o}~\cite{Szego75}.

The extremal property itself supplies the separation needed for rational
approximation, including at the endpoints.

\begin{lemma}[Polynomial Separation of Fekete Points]\label{t:lem:fekete-separation}
The minimum gap $\sigma=\min_{i<N}(x_{i+1}-x_i)$ satisfies
$\sigma\ge\underline\sigma:=1/[2(N-1)^2]$.
\end{lemma}
\begin{proof}
Replacing $x_i$ by $z\in[0,1]$ multiplies the Vandermonde product in
absolute value by $|\ell_i^x(z)|$. Optimality therefore gives
$\|\ell_i^x\|_\infty\le1$.
Markov's polynomial inequality~\cite{shadrin04} on $[0,1]$
gives $\| (\ell_i^x)'\|_\infty\le2(N-1)^2$.
Since $\ell_i^x(x_i)=1$ and $\ell_i^x(x_j)=0$ for $j\ne i$,
the mean-value theorem proves the claim.
\end{proof}

With this explicit gap bound, standard real-root isolation~\cite{BPR06}
computes the nodes to any polynomially small error in polynomial time.

\paragraph{Central Spacing.}

The multiplication construction uses only a fixed central band.  In that band, the following self-contained comparison shows that consecutive Fekete points have spacing $\Theta(1/N)$.

\begin{lemma}[Prüfer Comparison and Zero Count]\label{t:lem:sturm}
Let $v$ be a nonzero solution of
$v''+q(t)v=0$
on an interval $I$, where $0<a^2\le q(t)\le b^2$.  Then:
\begin{enumerate}[label=\textup{(\roman*)},leftmargin=2.6em]
\item the distance between two consecutive zeros whose intervening interval is contained in $I$ lies in $[\pi/b,\pi/a]$;
\item every subinterval of $I$ of length $L$ contains at least $aL/\pi-1$ zeros, where a negative lower bound is interpreted as vacuous.
\end{enumerate}
\end{lemma}

\begin{proof}
Because a nonzero solution and its derivative never vanish simultaneously, for any fixed $\kappa>0$ there are continuous functions $r_\kappa(t)>0$ and $\psi_\kappa(t)$ such that
\[
 v(t)=r_\kappa(t)\sin\psi_\kappa(t),
 \qquad
 v'(t)=\kappa r_\kappa(t)\cos\psi_\kappa(t).
\]
Differentiating these identities and using $v''=-qv$ gives the Prüfer equation
\begin{equation*}
 \psi_\kappa'(t)
 =\kappa\cos^2\psi_\kappa(t)
  +\frac{q(t)}{\kappa}\sin^2\psi_\kappa(t).
\end{equation*}
For $\kappa=a$, the lower bound $q\ge a^2$ implies $\psi_a'\ge a$.  For $\kappa=b$, the upper bound $q\le b^2$ implies $\psi_b'\le b$.  A zero of $v$ is exactly a point where either phase belongs to $\pi\mathbb Z$.  Between consecutive zeros the phase increases by exactly $\pi$.  Integrating the two derivative inequalities proves part~\textup{(i)}.

On a subinterval of length $L$, the $a$-phase increases by at least $aL$.  An interval of phase length $D$ contains at least $D/\pi-1$ multiples of $\pi$ after allowing for the two partial end intervals.  This proves part~\textup{(ii)}.
\end{proof}

\begin{lemma}[Central Spacing of Legendre--Gauss--Lobatto Nodes]\label{t:lem:central-spacing}
For every $N\ge N_0:=6592$:
\begin{enumerate}[label=\textup{(\roman*)},leftmargin=2.6em]
\item at least $N/7$ exact Fekete points lie in $[1/4,3/4]$;
\item if two consecutive exact Fekete points lie in $[1/4,3/4]$, then
$\frac1N\le x_{i+1}-x_i\le\frac4N$.
\end{enumerate}
Thus one may take $a_0=1$, $A_0=4$, and $c_0=1/7$.
\end{lemma}

\begin{proof}
The interior unshifted Legendre--Gauss--Lobatto (LGL) nodes are the zeros of $P_{N-1}'$, equivalently of the Jacobi polynomial $P_m^{(1,1)}$ with $m=N-2$; see Szeg\H{o}~\cite{Szego75}.  Put $t=\cos\vartheta$ and
$y(\vartheta)=P_m^{(1,1)}(\cos\vartheta)$.
The Jacobi differential equation becomes
$y''+3\cot\vartheta\,y'+m(m+3)y=0$.
After removing the first derivative by setting
$v(\vartheta)=(\sin\vartheta)^{3/2}y(\vartheta)$,
one obtains
\begin{equation*}
 v''+q(\vartheta)v=0,
 \qquad
 q(\vartheta)=\nu^2-\frac{3}{4\sin^2\vartheta},
 \qquad
 \nu=m+\frac32=N-\frac12.
\end{equation*}
On $I=[\pi/3,2\pi/3]$,
$\nu^2-1\le q(\vartheta)\le\nu^2-\frac34$.
Let
$a=\sqrt{\nu^2-1}$ and $b=\sqrt{\nu^2-3/4}$
For $N\ge2$, $a\ge N/2$ and $b<N$.  By \Cref{t:lem:sturm}, consecutive angular zeros contained in $I$ have gap at least $\pi/b>\pi/N$ and at most $\pi/a\le2\pi/N$.  Under the shift
$x(\vartheta)=\frac{1+\cos\vartheta}{2}$,
we have
$\frac{\sqrt3}{4} \le |x'(\vartheta)| \le\frac12$ for $\vartheta\in I$.
The mean-value theorem therefore gives
\[
 x_{i+1}-x_i
 \ge\frac{\sqrt3\pi}{4N}>\frac1N,
 \qquad
 x_{i+1}-x_i
 \le\frac\pi N<\frac4N.
\]

For the count, part~\textup{(ii)} of \Cref{t:lem:sturm} applied to the interval $I$, whose length is $\pi/3$, gives at least
$\frac{a}{3}-1\ge\frac N6-1$
zeros.  For $N\ge42$ this is at least $N/7$.  Zeros in $I$ correspond exactly to shifted points in $[1/4,3/4]$.  The larger displayed value of $N_0$ is chosen for later arrangement bounds and does not weaken the conclusion.
\end{proof}

\paragraph{Rational Near-Fekete Nodes.}

Fix a rational target localization radius $\rho_0>0$, instantiated in \Cref{t:sec:compiler-correctness}.  Choose rational ordered nodes
$0<\delta_1<\cdots<\delta_N<1$
with
\begin{equation}\label{t:eq:tag-approx}
 \max_i|\delta_i-x_i|\le\alpha,
\end{equation}
where
\begin{equation}\label{t:eq:alpha-conditions}
 \alpha\le
 \min\left\{
 \frac{\underline\sigma}{4N},
 \frac{\rho_0}{4N},
 \frac{\underline\sigma\rho_0^2}{128N^4},
 \frac{a_0}{20N},
 \frac1{100}
 \right\}.
\end{equation}
We choose a dyadic $\alpha$ between one quarter and one half of the displayed
minimum, and round the interior roots on a common dyadic grid of mesh
$\alpha/4$. Since $\rho_0=N^{-O(1)}$, this grid has polynomial
denominator and each node has $O(\log N)$ encoding length.
For definiteness choose the endpoint approximations
$\delta_1=\frac\alpha2$, $\delta_N=1-\frac\alpha2$,
and obtain the interior roots from isolating intervals of width at most
$\alpha/4$ before rounding. For every adjacent pair,
$\delta_{i+1}-\delta_i \ge (x_{i+1}-x_i)-2\alpha \ge\sigma/2 \ge\underline\sigma/2;$
the endpoint pairs satisfy the same bound because their approximation errors are at most $\alpha$.  Ordering also implies that every node is at least $\delta_1=\alpha/2$ and at most $\delta_N=1-\alpha/2$.  Thus every node and every node gap has polynomial encoding length, and
$g:=\min_{i<N}(\delta_{i+1}-\delta_i) \ge \frac{\underline\sigma}{2}$, $b:=\min_i\{\delta_i,1-\delta_i\}=\frac\alpha2=N^{-O(1)}$.

Define the root coordinate
$\zeta_i:=-\frac{1-\delta_i}{\delta_i}=1-\frac1{\delta_i}$.
Let $I_{\mathrm{cen}}$ be the consecutive set of indices for which $x_i\in[1/4,3/4]$.  For these indices, $|\delta_i-x_i|\le1/100$, so $\delta_i\in[1/5,4/5]$.  The map
$\zeta(\delta)=1-\frac1\delta$
has derivative $1/\delta^2\in[25/16,25]$ on this interval.

\begin{corollary}[Central Rational-Root Spacing]\label{t:cor:zeta-spacing}
For every $N\ge N_0$, the set $I_{\mathrm{cen}}$ has at least $N/7$ indices.  For every consecutive pair $i,i+1\in I_{\mathrm{cen}}$,
$\frac1N \le \zeta_{i+1}-\zeta_i \le \frac{103}{N}$.
Thus one may take $a_\zeta=1$, $A_\zeta=103$, and $c_\zeta=1/7$.
\end{corollary}

\begin{proof}
The count is \Cref{t:lem:central-spacing}.  For two consecutive central indices, \Cref{t:lem:central-spacing,t:eq:alpha-conditions} gives
$\frac{9}{10N} \le \delta_{i+1}-\delta_i \le \frac{41}{10N}$,
because perturbing both endpoints changes a gap by at most $2\alpha\le1/(10N)$.  The mean-value theorem for $\zeta(\delta)$ now gives
$\zeta_{i+1}-\zeta_i \ge\frac{25}{16}\frac{9}{10N}>\frac1N$,
and
$\zeta_{i+1}-\zeta_i \le25\frac{41}{10N}<\frac{103}{N}$.
\end{proof}

\begin{lemma}[Calibration Products at Near-Fekete Nodes]\label{t:lem:fekete-products}
For even $N$ and the nodes above, $c_*\ge2^{1-2N}$.
\end{lemma}
\begin{proof}
Put $m=N-1$. The standard Legendre identities~\cite{Szego75}
give the energy
\[
 \mathcal E(t)=P_m(t)^2+
 \frac{1-t^2}{m(m+1)}P_m'(t)^2,
 \qquad
 \mathcal E'(t)=\frac{2t}{m(m+1)}P_m'(t)^2.
\]
Thus $\mathcal E(t)\ge\mathcal E(0)$.
Since $m$ is odd, $|P_m'(0)|=m\binom{m-1}{(m-1)/2}/2^{m-1}\ge1$;
the central binomial coefficient is at least the average of its row.
At a Lobatto node, either $P_m'=0$ or $t=\pm1$, so
$|P_m(t)|\ge1/\sqrt{m(m+1)}\ge1/N$.

Let $a_m=2^{-m}\binom{2m}{m}\le2^m$ be the leading coefficient.
The monic polynomial with roots $x_i$ is
$(t^2-1)P_m'(t)/(m a_m2^N)$, with $t=2z-1$.
The Legendre equation therefore gives
$\prod_{j\ne i}|x_j-x_i|=N|P_m(2x_i-1)|/(2^m a_m)\ge4^{-m}$.
Rounding changes each gap by at most $2\alpha$. Since
$\alpha\le\underline\sigma/(4N)$,
the product loses at most the factor
$(1-2\alpha/\underline\sigma)^{N-1}\ge1/2$.
This proves $c_*\ge4^{1-N}/2$.
\end{proof}

\paragraph{The Vandermonde Gap.}

The Vandermonde maximum must be separated from configurations whose
coordinates move a prescribed distance. We first quantify the loss in
$\log V$, then transfer the estimate to the rational nodes used in the game.

\begin{lemma}[Strong Vandermonde Gap]\label{t:lem:strong-vandermonde}
Let $x$ be the Fekete configuration of $[0,1]$.  For every ordered $y\in[0,1]^N$,
\begin{equation}\label{t:eq:strong-vandermonde}
 \log V(y)
 \le
 \log V(x)-\frac12\|y-x\|_2^2,
\end{equation}
with $\log V(y)=-\infty$ when $y$ has a collision.
\end{lemma}

\begin{proof}
Assume $y$ is strictly ordered and put $v=y-x$.  By the optimality conditions from \Cref{t:lem:fekete-characterization}, every interior derivative of $\Phi=\log V$ vanishes at $x$, while the endpoint derivatives have signs opposite to $v_1\ge0$ and $v_N\le0$.  Hence
$\nabla\Phi(x)\cdot v\le0$.
Along $z(t)=x+tv$,
\[
 -\frac{d^2}{dt^2}\Phi(z(t))
 =\sum_{i<j}\frac{(v_j-v_i)^2}{(z_j(t)-z_i(t))^2}
 \ge\sum_{i<j}(v_j-v_i)^2.
\]
The identity
\[
 \sum_{i<j}(v_j-v_i)^2
 =N\|v\|_2^2-\Bigl(\sum_i v_i\Bigr)^2
\]
and the endpoint sign condition imply
$\Bigl(\sum_i v_i\Bigr)^2\le(N-1)\|v\|_2^2$.
Indeed, if $\sum_i v_i\ge0$, discard the nonpositive term $v_N$ and apply Cauchy--Schwarz to the remaining $N-1$ coordinates; the negative-sum case is symmetric using $v_1\ge0$.  Thus the second derivative is at most $-\|v\|_2^2$.  Integrating twice gives \cref{t:eq:strong-vandermonde}.  The collision case follows by continuity with $\log V=-\infty$.
\end{proof}

\begin{lemma}[Near-Fekete Product Gap]\label{t:lem:near-fekete-gap}
Assume \crefrange{t:eq:tag-approx}{t:eq:alpha-conditions} and let $w\in[0,1]^N$. If some $k$ satisfies
$|w_k-\delta_k|\ge\rho_0$,
then
\begin{equation}\label{t:eq:near-fekete-gap}
 \left|\prod_{i=1}^N\ell_i^\delta(w_i)\right|
 \le
 \exp\left(-\frac{\rho_0^2}{32N}\right).
\end{equation}
\end{lemma}

\begin{proof}
Let
$\varepsilon_V=\log V(x)-\log V(\delta)$.
Because every exact gap is at least $\sigma\ge\underline\sigma$ and $\alpha\le\underline\sigma/4$,
\[
 \varepsilon_V
 \le4\alpha\sum_{i<j}\frac1{x_j-x_i}
 \le\frac{2\alpha N^2}{\underline\sigma}.
\]
For fixed $i$ and $z\in[0,1]$, let $y^{(i,z)}$ be the sorted configuration obtained by replacing $\delta_i$ by $z$.  Vandermonde cancellation gives
\[
 \log|\ell_i^\delta(z)|
 =\log V(y^{(i,z)})-\log V(\delta)
 \le \varepsilon_V-\frac12\|y^{(i,z)}-x\|_2^2.
\]
In particular every such logarithm is at most $\varepsilon_V$.

Take $i=k$ and $z=w_k$.  Sorting does not change the sum of coordinates, so
$\left|\sum_j y_j^{(k,w_k)}-\sum_j x_j\right| \ge\rho_0-N\alpha \ge\frac{3\rho_0}{4}$.
By Cauchy--Schwarz,
$\|y^{(k,w_k)}-x\|_2^2 \ge\frac{\rho_0^2}{2N}$.
Therefore
$\log\left|\prod_i\ell_i^\delta(w_i)\right| \le N\varepsilon_V-\frac{\rho_0^2}{4N}$.
The third term in \cref{t:eq:alpha-conditions} implies $N\varepsilon_V\le\rho_0^2/(64N)$, which is stronger than \cref{t:eq:near-fekete-gap}.
\end{proof}

\subsection{Payoff Bounds for Circuit Simulation}\label{t:sec:sparse-details}

The same rational nodes must support both exact recovery identities and
the multiplication construction. We bound the recovery tables first,
then arrange the active players so that the unwanted terms in each
second-order statistic are small.

\paragraph{Magnitude and Encoding Length of the Payoff Entries.}

The recovery identities are exact at the prescribed nodes. We next check that
the corresponding rational payoff tables have polynomial-size descriptions.

\begin{lemma}[Bounds on the Payoff Tables]\label{t:lem:reader-bits}
The tables $T^{(1)}_{o,j}$ and $T^{(2)}_{o,j}$ used by the construction have rational entries of
$O(N\log N)$ encoding length and magnitude at most $2^{O(N)}$.
The complete tables are constructible in polynomial time.
\end{lemma}

\begin{proof}
All root players used in a gate table are active, so
$\delta_j\in[1/5,4/5]$ and $|\zeta_j|\le4$.
Factoring $A_\ell(\zeta_j)=(\delta_j-\delta_\ell)/\delta_j$ gives
$|N_{o,j}|^{-1}=\delta_j^{N-3}|\delta_j-\delta_o|/|c_j|
\le c_*^{-1}\le2^{2N-1}$ by \Cref{t:lem:fekete-products}.
Each numerator is bounded by $N^2 4^N$, giving the exponential magnitude
bound. Every node has $O(\log N)$ encoding length, and each entry uses
$O(N)$ rational factors, giving $O(N\log N)$ encoding length.
There are only $O(N^2)$ histograms per observer.
\end{proof}

\paragraph{Separating the Active Players.}

Let the source circuit contain $G$ gates and $M\le G$ bilinear gates.  We may assume $G\ge1$; an empty source instance can be mapped to any fixed one-player game.  The gate construction below uses six ordinary players per gate and four gate-specific auxiliary players per bilinear gate.  Let $\calA$ denote the set of players participating in the gate equations.  Its cardinality is
$A:=|\calA|=6G+4M\le10G$.
All other calibrated players are \emph{fillers}.

Fix universal constants
\begin{equation}\label{t:eq:layout-constants}
 C_{\mathrm{sp}}:=1024\,\frac{A_\zeta}{a_\zeta},
 \qquad
 C_N:=\frac{2}{c_\zeta}.
\end{equation}
The explicit values from \Cref{t:cor:zeta-spacing} give one concrete choice; their numerical magnitudes play no conceptual role below.
Choose
$L=C_{\mathrm{sp}}\left\lceil\log_2(2A)\right\rceil$
which is an explicitly computable integer and is at least $C_{\mathrm{sp}}\log(2A)$. Choose the total calibrated population $N$ to be the least even integer larger than
\begin{equation}\label{t:eq:N-choice}
 C_NA(L+2)+N_0.
\end{equation}
Thus $N=\Theta(G\log(2G))$.  There are at most $A$ active blocks, so an arrangement with $L$ fillers between consecutive blocks uses fewer than $A(L+2)$ positions.  By \Cref{t:cor:zeta-spacing} and \cref{t:eq:N-choice}, the central index set contains at least
$c_\zeta N>2A(L+2)$
positions.  Hence it contains enough consecutive rational nodes for the entire arrangement.

Place the active players in blocks along these consecutive central nodes.  Every active player forms a singleton block except that, for each bilinear gate, the two designated auxiliary players $V_i,B_i$ form one adjacent block.  Separate distinct blocks by at least $L$ filler nodes.  The remaining auxiliary players $C_i,W_i$ are singleton blocks and are therefore $L$-separated from the $V_i,B_i$ pair and from every other active block.  All unused positions, including the near-endpoint nodes, are fillers.

For a bilinear gate $i$, orient the adjacent pair so that
$h_i:=\zeta_{B_i}-\zeta_{V_i}>0$ and $d_i:=h_i^{-1}$.

\begin{lemma}[Bounding the Unwanted Terms]\label{t:lem:sparse}
With the constants in \cref{t:eq:layout-constants}, every bilinear pair satisfies
\begin{equation}\label{t:eq:sparse-bound}
 h_i\le\frac1{64},
 \qquad
 \Sigma_i:=
 \sum_{p\in\calA\setminus\{B_i,V_i\}}
 \left|\frac{h_i}{\zeta_{B_i}-\zeta_p}\right|
 \le\frac1{64}.
\end{equation}
Moreover $d_i=\Theta(N)$, with absolute constants independent of the gate.
\end{lemma}

\begin{proof}
All active nodes lie in the central band, so \Cref{t:cor:zeta-spacing} gives
$\frac{a_\zeta}{N} \le\zeta_{q+1}-\zeta_q \le\frac{A_\zeta}{N}$
for consecutive central positions.  Hence
$\frac{a_\zeta}{N}\le h_i\le\frac{A_\zeta}{N}$,
which gives $d_i=\Theta(N)$.  Since $N>N_0=64A_\zeta$, it also gives $h_i\le A_\zeta/N<1/64$.

Consider blocks at block-distance $k\ge1$ from the block containing $B_i,V_i$.  Since distinct blocks have at least $L$ fillers between them and each block has size at most two, every active point in such a block is at index distance at least $kL/2$ from $B_i$ once $L\ge4$.  Therefore
$\left|\frac{h_i}{\zeta_{B_i}-\zeta_p}\right| \le \frac{2A_\zeta}{a_\zeta}\frac1{kL}$.
At each block-distance there are at most two blocks, one on either side, and each block has at most two active players.  Thus
\[
 \Sigma_i
 \le
 \frac{8A_\zeta}{a_\zeta L}
 \sum_{k=1}^{A}\frac1k
 \le
 \frac{8A_\zeta(1+\log A)}{a_\zeta L}.
\]
For $A\ge1$, $1+\log A\le2\log(2A)$, and $L\ge C_{\mathrm{sp}}\log(2A)$.  Therefore
$\Sigma_i \le\frac{16A_\zeta}{a_\zeta C_{\mathrm{sp}}} =\frac1{64}$.
This completes both estimates in \cref{t:eq:sparse-bound}.
\end{proof}

Every filler receives the constant gate payoff table
$d_F(h)=-1$.
At the accuracies used below, this forces $\theta_F=0$ without any recovery statistic.

\subsection{Approximation Error and Exact Equilibria}\label{t:sec:compiler-correctness}

The main correctness proof is in \Cref{t:sec:gate-overview}. We give explicit choices for its scales and the final decoding estimate.

We choose the geometry before the payoff scale.  For the even population $N$ from \cref{t:eq:N-choice}, compute the separation lower bound $\underline\sigma$ and put
$g_0=\underline\sigma/2$, $\rho_0=g_0^2/(8N^3)$.
Choose the dyadic $\alpha$ and rational nodes as in \Cref{t:sec:fekete}.
Their minimum gap is at least $g_0$, so $\rho_0\le g^2/(4N^3)$,
as required by \Cref{t:lem:local-inverse}.
Both $\rho_0^{-1}$ and $\alpha^{-1}$ are polynomial in $N$.

After the arrangement and all histogram tables have been constructed,
choose a power of two $B\ge1$ that bounds their magnitudes, including every
gate payoff table $d_P$ and product table $\Pi_i$.
By \Cref{t:lem:reader-bits}, we may take $B\le2^{O(N)}$.
Set $d_{\max}=\max(\{1\}\cup\{d_i:i\text{ bilinear}\})$. Use a single rational bound
\begin{equation}\label{t:eq:common-bound}
 H=\max\{64,N,B,d_{\max},b^{-1},c_*^{-1},\rho_0^{-1}\}.
\end{equation}
The geometric reciprocals are polynomial in $N$, while
$B,c_*^{-1}\le2^{O(N)}$. Hence $H\le2^{O(N)}$.
The following choice leaves ample slack in every estimate:
$\eta=H^{-16}$, $\xi_*=H^{-20}$ and $\Lambda=H^{19}$.
Only their inverse-exponential size matters; the particular constant exponents are unimportant.

The scale choice therefore gives $\xi_*\ge2^{-O(N)}$ and
$\Lambda\le2^{O(N)}$. We verify the accuracy below before bounding
the rational descriptions.

For $0\le\xi\le H^{-20}$ and $\eta=H^{-16}$, the bounds needed in the
main proof follow from \Cref{t:eq:common-bound}. In particular,
$\gamma+\xi\le2H^{-15}$, whereas the right-hand side of
\Cref{t:eq:global-gamma} is at least $H^{-7}/512$. Thus global localization
applies. The remaining estimates are summarized in \Cref{tab:three-strategy-scales}.

\begin{table}[t]
\centering
\renewcommand{\arraystretch}{1.2}
\begin{tabularx}{\linewidth}{@{}l l X@{}}
\toprule
Quantity & Bound & Use in the Proof \\
\midrule
Initial signal displacement & $8H^{-11}<b/2$ & Every signal probability is interior. \\
Error in a gate expectation & $16H^{-9}$ & Transfer expectations to the prescribed nodes. \\
Initial gate error $e_0$ & $2H^{-4}$ & Establish support using \Cref{t:lem:semantic-stability}. \\
Refined gate error $e(\xi)$ & $2H^{16}\xi$ & Reapply the same stability lemma. \\
Projected gate residual & $24H^{17}\xi$ & Recover the source equations. \\
\bottomrule
\end{tabularx}
\caption{Explicit Error Bounds for the Three-Strategy Reduction. Only the last
two bounds are used after the supports have been established.}
\label{tab:three-strategy-scales}
\end{table}

The initial expectation error is at most $2BN(8H^{-11})\le16H^{-9}$.
Consequently $e_0\le H^{-4}+16H^{-9}\le2H^{-4}$, which is below
$1/[100(1+d_{\max})]$. Once $s,t$ are supported, the local inverse gives
\Cref{t:eq:main-fine-w}, and
$e(\xi)\le(H^{16}+4H^4)\xi\le2H^{16}\xi$.
The second application of the stability lemma therefore yields
\begin{equation}\label{t:eq:robust-projection}
 |u_i-\mathcal P(f_i(u))|\le24H^{17}\xi
 =:\varepsilon_{\mathrm{gate}}.
\end{equation}

\paragraph{Decoding the Circuit Values.}

Define
\begin{equation*}
 \overline u_i=\mathcal P(u_i),\qquad
 \overline x_i=2\overline u_i-\frac12,\qquad
 \widehat x_i=\clip_{[0,1]}\left(
 \frac{2\sigma_{X_i,c}}{\delta_{X_i}}-\frac12\right).
\end{equation*}
By \cref{t:eq:robust-projection}, both $|u_i-\overline u_i|$ and
$|\overline u_i-\mathcal P(f_i(u))|$ are at most $\varepsilon_{\mathrm{gate}}$.
For complement and truncated doubling, clipping commutes with the protected operation, so their decoded residuals are at most $2\varepsilon_{\mathrm{gate}}$.
For the bilinear expression, each partial derivative is at most $5/8$ on $[0,1]^2$.  Clipping its two inputs changes its value by at most $5\varepsilon_{\mathrm{gate}}/4$; at clipped inputs the value lies in $[3/8,3/4]$.  Nonexpansiveness of $\mathcal P$ and rescaling by two therefore bound its decoded residual by $9\varepsilon_{\mathrm{gate}}/2$.  Thus $\overline x$ satisfies every source gate within $5\varepsilon_{\mathrm{gate}}$.

Finally, $\sigma_{X_i,c}=w_{X_i}u_i$ and clipping is nonexpansive, so
\begin{equation*}
 |\widehat x_i-\overline x_i|
 \le\frac2b|w_{X_i}-\delta_{X_i}|
 \le\frac{4N}{bc_*}\xi\le4H^3\xi.
\end{equation*}
Changing every source coordinate by $D$ changes any gate residual by at most $3D$.  The final error is therefore at most
$(120H^{17}+12H^3)\xi\le H^{19}\xi=\Lambda\xi$.
At $\xi=0$ all errors vanish, $w=\delta$, and $u_i\in[1/4,3/4]$.  Clipping is inactive and the decoder is exactly the affine coordinate map \cref{t:eq:main-SL-decoder}.

\paragraph{Polynomial Rational Descriptions.}

All quantities are rational expressions in polynomially many node factors. The following lemma collects the resulting size bounds.

\begin{lemma}[Polynomial-Size Construction]\label{t:lem:size}
The reduction is computable in polynomial time. It uses
$N=O(G\log(2G))$ players and rational payoff entries of
$O(N^2\log N)$ encoding length. The parameters satisfy $H\le2^{O(N)}$,
$\xi_*\ge2^{-O(N)}$, and $\Lambda\le2^{O(N)}$.
\end{lemma}
\begin{proof}
We use six players per source gate and four additional players per bilinear
gate. The separation among central nodes requires $O(\log(2G))$ fillers per
group. Hence the population has the stated order, and each player has
$3\binom{N+1}{2}=O(N^2)$ payoff entries.

We isolate the roots of $(1-t^2)P_{N-1}'(t)$ to polynomially small error using
standard real-algebraic algorithms~\cite{BPR06}, then round on a grid of
polynomial denominator. Every recovery entry has $O(N\log N)$ encoding
length by \Cref{t:lem:reader-bits}. A gate table sums at most $O(N)$ such
entries with coefficients of $O(\log N)$ encoding length, so
$O(N^2\log N)$ suffices. Calibration products use $O(N)$ node factors,
and the fixed powers defining the final scales preserve these bounds.
We enumerate all histograms and evaluate their payoffs in polynomial time.
Rounding the maximum magnitude upward to a power of two gives $B$.
The exponential magnitude bounds follow from
\Cref{t:lem:fekete-products,t:lem:reader-bits}.
\end{proof}

\paragraph{Payoff Range.}\label{t:lem:normalization}
Every base entry lies in $[-2,2]$. The gate perturbation has magnitude at most
$\eta B\le H^{-15}<1/2$, so the unnormalized entries lie in $[-5/2,5/2]$.
The affine normalization used in the main proof therefore puts them in $[0,1]$.

\printbibliography[heading=bibintoc,title={References}]
\end{document}